\documentclass[a4paper,11pt,reqno]{amsart}
\usepackage[utf8]{inputenc}
\usepackage[T1]{fontenc} 

\usepackage{amsthm,amsmath,amsfonts,amssymb,dsfont,bm,amstext,amsopn,mathrsfs,mathtools,esint}
\usepackage[colorlinks, linkcolor={blue}, citecolor={red}]{hyperref}
\usepackage[hmargin=2.5cm,vmargin=2.5cm]{geometry}
\usepackage[dvipsnames]{xcolor}
\usepackage{tikz}
\usetikzlibrary{calc,decorations.pathreplacing,patterns,arrows.meta}

\usepackage{notations}
\usepackage{comment}

\usepackage[backend=biber,
style=alphabetic,
maxbibnames=99,
doi=false, url=false, giveninits=true, date=year]{biblatex}
\usepackage{csquotes}

\title[Kohn--Sham models for encapsulated two-dimensional materials]{Kohn--Sham models for encapsulated two-dimensional materials}

\author{Éric Cancès}
\address[Éric Cancès]{CERMICS, CNRS, ENPC, Institut Polytechnique de Paris, Marne-la-Vallée, France, and MATHERIALS team-project, Inria Paris, France.}
\email{eric.cances@enpc.fr}

\author{David Gontier}
\address[David Gontier]{CERMICS, CNRS, ENPC, Institut Polytechnique de Paris, Marne-la-Vallée, France, and MATHERIALS team-project, Inria Paris, France.} \email{david.gontier@enpc.fr}

\author{Solal Perrin-Roussel}
\address[Solal Perrin-Roussel]{UMPA, UMR 5669 CNRS, École Normale Supérieure de Lyon, 46 allée d’Italie, 69364 Lyon Cedex 07, France.}
\email{solal.perrin\_roussel@ens-lyon.fr}

\begin{document} 
    
    \begin{abstract}
        We study Kohn--Sham Density Functional Theory (DFT) models describing the electronic structure of two-dimensional materials placed in a three-dimensional environment, encapsulated between two parallel conducting electrodes. In this geometry, the Dirichlet boundary conditions at the electrodes screen the Coulomb interaction, which becomes effectively short-ranged, of Yukawa type. We prove that some nonlinear Kohn--Sham DFT models are well-posed in this setting, both for periodic materials (such as graphene) and for quasi-periodic materials (such as twisted bilayer graphene and other moir\'e materials for generic incommensurate twist angles).
    \end{abstract}

        \maketitle
    \tableofcontents

    \section{Introduction}

    The study of 2D and moir\'e materials is one of the most vibrant fields in today's condensed matter physics. It is also a rich source of exciting mathematical and numerical problems. 
    The physics of 2D materials blossomed after the seminal paper by Geim and Novoselov reporting experimental transport measurements on graphene~\cite{novoselov2004electric}, made possible by the famous scotch-tape exfoliation technique used to isolate monolayer graphene. In the following years, new experimental techniques were developed to exfoliate or grow high-quality monolayer graphene, as well as many other single-layer 2D materials, such as hexagonal boron nitride (hBN) or transition metal dichalcogenides (TMDs)~\cite{tscheppe2024magnetism,miao2021strong,zang2022dynamical}. 
    In parallel, multilayer 2D materials were synthesized by stacking several identical or different monolayers on top of one another, in aligned or twisted configurations. Lattice mismatches in hetero-bilayers (e.g. graphene on top of hBN), or rotation angles in twisted homo-bilayers (e.g. twisted bilayer graphene - TBG), lead to the appearance of moir\'e patterns, hence the name moiré materials.

    \medskip
    
    For small lattice mismatches or small twist angles, the apparent moiré periodicity is much larger than the period of the underlying monolayer periodic structures and gives rise to new physical phenomena at the nanoscale.
    What makes 2D moiré materials so special is that, in contrast to bulk 3D materials or surfaces, the electronic structure of a 2D moiré material can be finely controlled, to some extent, by tuning different parameters: twist angle, mechanical strain, longitudinal and/or transverse electric and/or magnetic fields.
    In particular, the filling factor, that is the number of extra or missing electrons per moiré cell with respect to charge neutrality, can be controlled by placing the moiré material between two metallic electrodes and imposing a potential difference between these electrodes (a procedure called {\em gating}). To experimentally realize gating, the 2D material under study must first be encapsulated between two thin insulating layers (e.g. several layers of hBN) in order to avoid short circuits (see Fig.~\ref{fig:encapsulated}).
    
\begin{figure}[ht]
    \centering
    \begin{tikzpicture}[scale=1, >=Stealth]
        \def\L{8}   \def\tins{2.0}  
        \draw[->] (5.0,-5.0) -- (5.0,5.0) node[above] {\(z\)};

        \draw (5.0, { \L/2 }) -- ++(0.15,0);
        \node[right] at (5.2, { \L/2 }) {\(+\frac{L}{2}\)};
        \draw (5.0, { 0 }) -- ++(0.15,0);
        \node[right] at (5.2, { 0 }) {\(0\)};
        \draw (5.0, { -\L/2 }) -- ++(0.15,0);
        \node[right] at (5.2, { -\L/2 }) {\(-\frac{L}{2}\)};

        \fill[gray!30] (-3.8,\L/2+0.12) rectangle (3.8,\L/2-0.12);
        \draw[thick] (-3.8,\L/2+0.12) rectangle (3.8,\L/2-0.12);
        \node at (0,\L/2+0.5) {\(V=v_0/2\)};
        
        \fill[gray!30] (-3.8,-\L/2+0.12) rectangle (3.8,-\L/2-0.12);
        \draw[thick] (-3.8,-\L/2+0.12) rectangle (3.8,-\L/2-0.12);
        \node at (0,-\L/2-0.5) {$V=-v_0/2$};
        
        \fill[blue!8,draw=blue!30] (-3.4,\tins/2+0.6) rectangle (3.4,\L/2-0.16);
        \node at (0,\L/2 - \tins/2 - 0.3) {Insulator (hBN)};
        
        \fill[blue!8,draw=blue!30] (-3.4,-\L/2+0.16) rectangle (3.4,-\tins/2-0.6);
        \node at (0,-\L/2 + \tins/2 + 0.3) {Insulator (hBN)};

        
        \fill[red!30] (-3.8,0.3) rectangle (3.8,0.1);
        \draw[thick] (-3.8,0.3) rectangle (3.8,0.1);
        \fill[red!30] (-3.8,-0.1) rectangle (3.8,-0.3);
        \draw[thick] (-3.8,-0.1) rectangle (3.8,-0.3);
        
        \fill[red!15] (-3.8,0.32) rectangle (3.8,0.08);
        \draw[thick]  (-3.8,0.32) rectangle (3.8,0.08);
        \fill[red!15] (-3.8,-0.08) rectangle (3.8,-0.32);
        \draw[thick]  (-3.8,-0.08) rectangle (3.8,-0.32);
        
        \foreach \x in {-3.4,-2.7,-2.0,-1.3,-0.6,0.1,0.8,1.5,2.2,2.9}
        { \fill[red!80] (\x,0.20) circle (0.07); }
        \foreach \x in {-3.05,-2.35,-1.65,-0.95,-0.25,0.45,1.15,1.85,2.55,3.25}
        { \fill[red!80] (\x,-0.20) circle (0.07); }
        \node at (0,0.65) {moiré material (TBG)};
        
    \end{tikzpicture}
    
    \caption{Schematic diagram of the physical system.}
    \label{fig:encapsulated}
\end{figure}
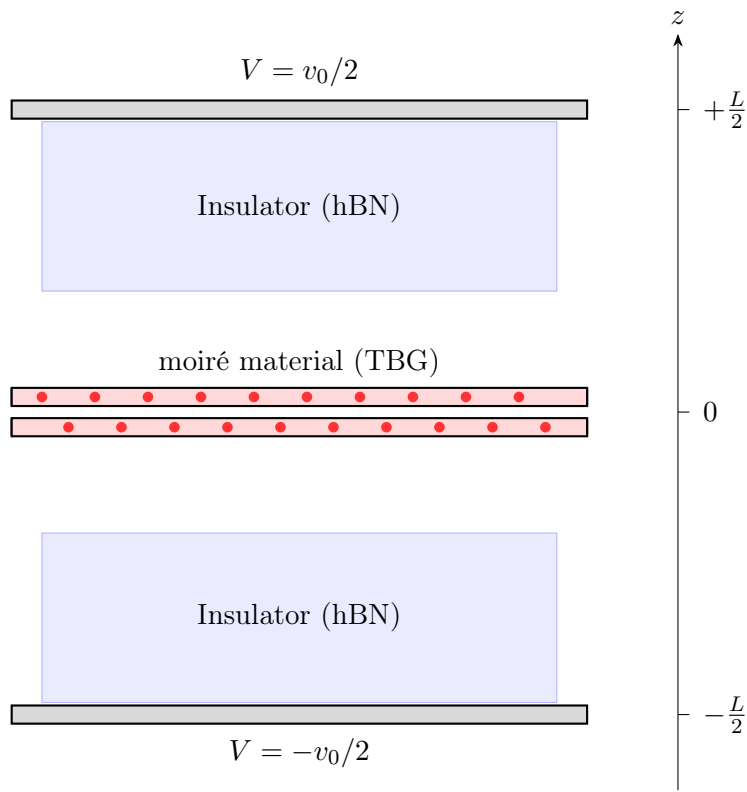

\medskip
    
    A theoretical breakthrough in the study of moiré materials was made in 2011 by Bistritzer and MacDonald~\cite{bistritzer2011moire,bistritzer2011moirea}, who predicted the presence of so-called magic angles in TBG, at which strongly correlated electronic effects should emerge. 
    This prediction was confirmed by experiments conducted in 2018 in Jarillo-Herrero's group~\cite{cao2018correlated, cao2018unconventional}, where strongly correlated insulating phases and, even more interestingly, {\em unconventional} superconducting phases (which cannot be explained by the classical BCS theory of superconductivity), were observed in the phase diagram of twisted bilayer graphene close to the first magic angle $\theta \simeq 1.1^\circ$, for sufficiently low temperatures and suitable values of the filling factor. 
    This discovery generated enormous interest in the condensed-matter community, because it opened new perspectives for understanding unconventional superconductivity and, ultimately, the still not fully understood high-temperature superconductivity of cuprates. Since then, correlated and superconducting phases have been observed in several other moir\'e materials, including twisted trilayer graphene (TTG)~\cite{park2021tunable}, twisted double bilayer graphene (TDBG)~\cite{liu2020tunable}, and moir\'e transition metal dichalcogenide (TMD) bilayers~\cite{wang2020correlated}, as well as in non-moir\'e multilayer 2D materials such as rhombohedral graphene~\cite{ZhoXieTan-21}.
    
    \medskip
    
    In this field, perhaps more than in any other area of physics, modeling and experiment are closely intertwined. Throughout the development of the field, the modeling of 2D and moir\'e materials has provided indispensable guidance for experimental studies, while at the same time being necessary to understand the rich and unusual electronic properties of these materials.
    
    \medskip
    
    Long before monolayer graphene was experimentally realized, the electronic structure of graphene was investigated in 1947 by Wallace~\cite{wallace1947band}.
    Using a minimal model describing non-interacting electrons on a hexagonal lattice, Wallace was able to predict the semimetallic character of monolayer graphene at charge neutrality, as well as the fact that, for low-energy excitations, charge carriers (electrons and holes) behave as massless relativistic particles with Fermi velocity $v_{\rm F}$ ($v_{\rm F} \approx 1 \times 10^6$ m/s for graphene, which is 300 times slower than the true speed of light). These properties follow from the presence of Dirac cones at the $K$ and $K'$ points of the band diagram of graphene. The Bistritzer--MacDonald (BM) model is a remarkably successful minimal effective non-interacting model describing low-energy electronic excitations in TBG at the moir\'e scale.
    The BM Hamiltonian $H_{\rm BM}$ is a periodic self-adjoint operator on $L^2(\R^2;\C^4)$ ($L^2(\R^2;\C^8)$ if both the $K$ and $K'$ valleys are included), with unit cell given by the (mesoscopic) moir\'e cell. It depends on the twist angle $\theta$ and only three empirical parameters: the Fermi velocity $v_{\rm F}$ and two parameters $w_0$ and $w_1$ describing interlayer interactions at AA and AB stacking respectively. Bistritzer and MacDonald observed the emergence of almost flat bands at the Fermi level in the band diagram of $H_{\rm BM}$ for specific ``magic'' twist angles.
    Flat bands at the Fermi level suggest the existence of partially occupied localized electronic states, and the possibility of strongly correlated insulating or superconducting phases. While non-interacting one-particle models such as BM have proven to be invaluable in guiding physical intuition, they are not sufficient to make quantitative predictions on the phase diagram itself. To push the analysis further and numerically compute an approximate phase diagram, it is necessary to construct effective {\em interacting} moir\'e-scale models going beyond the non-interacting ones, and to simulate these models using, for instance, Dynamical Mean-Field Theory~\cite{georges1996dynamical,rai2024dynamical,datta2023heavy,zang2022dynamical} (DMFT, see~\cite{cances2026mathematical} for some mathematical insights on this method).
    
    \medskip
    
    The modeling of 2D and moiré materials has also stimulated mathematical research. Regarding graphene, Wallace's non-interacting lattice model was extended in two directions.
    On the one hand, the existence of Dirac cones at $K$ and $K'$ was shown by Fefferman and Weinstein~\cite{fefferman2012honeycomb} (see also~\cite{berkolaiko2018symmetry}) to be a generic property of continuous 2D Schrödinger operators with honeycomb symmetries. On the other hand, Giuliani, Mastropietro and Porta~\cite{giuliani2010lattice, GiuMas-09} studied an {\em interacting} model of electrons on hexagonal lattices using exact renormalization group techniques. Other mathematical works focused on the modeling of TBG. The BM model was derived from atomic-scale non-interacting tight-binding models by Watson and Luskin~\cite{WatLus-21}, and a variant of it directly from Density Functional Theory (DFT) by Canc\`es, Garrigue and Gontier~\cite{cances2023simple}. The study of the BM model also gave rise to several mathematical works.
    Building on previous works by Vishwanath and collaborators~\cite{tarnopolsky2019origin}, the existence of magic angles for the chiral approximation of the BM model was established in a series of works by Becker, Zworski et al.~\cite{becker2022mathematics}, and Watson, Kong, MacDonald and Luskin~\cite{watson2023bistritzer}. Extending previous works by Zaletel and collaborators~\cite{bultinck2020ground}, Becker, Lin, and Stubbs~\cite{becker2025exact} proved that the ground states of some interacting Bistritzer--MacDonald (IBM) models at charge neutrality are Hartree--Fock states with specific symmetries.
    
    \medskip
    
    Among the wide variety of electronic structure models used in the study of 2D and moir\'e materials, Density Functional Theory (DFT)~\cite{hohenberg1964inhomogeneous,kohn1965selfconsistent} plays an important role.
    DFT alone is generally unable to quantitatively reproduce the phase diagrams of moir\'e materials, as its purpose is only to compute the ground-state energy and one-particle density, quantities that do not {\it a priori} provide any direct insight into correlation functions, from which strongly correlated effects can be inferred (note, however, that a DFT theory for BCS superconductors has been developed~\cite{oliveira1988densityfunctional,luders2005initio}).
    On the other hand, DFT is arguably the best model for quantitatively studying atomic relaxation~\cite{carr2018relaxation}, an effect reported to play an important role in the electronic properties of multilayer 2D materials~\cite{nam2017lattice,carr2019exact,carr2018pressure}, as well as phonon modes~\cite{lu2022lowenergy}, and for providing reasonable approximations of electron-phonon coupling parameters~\cite{baroni2001phonons}.
    It is also a valuable tool for parameterizing non-interacting and interacting models of 2D and moir\'e materials~\cite{fang2018electronic,cances2023simple}.
    
    \medskip
    
    In this article, we provide a mathematical analysis of Kohn--Sham models for encapsulated 2D and moir\'e materials (experimental setting sketched in Fig.~\ref{fig:encapsulated}) at the reduced Hartree--Fock (rHF), also called Hartree, or RPA (Random Phase Approximation), level of theory. On the Jacob's ladder of DFT approximations proposed by Perdew~\cite{perdew2001jacobs}, this setting corresponds to the ground level, for which the exchange-correlation functional is chosen to be zero.
    Although this model is not sufficiently accurate to obtain quantitative results, it is very interesting from a mathematical viewpoint, as its mathematical structure is similar to that of practically used Kohn--Sham models, with the additional property of being convex, which greatly simplifies its analysis. Note that the exact DFT model based on Lieb's density functional~\cite{levy1979universal,lieb1983density} is convex, and that non-convexities are in fact induced by (semi-)local approximations of the exchange-correlation functional.
    Within the rHF approximation, it is possible to prove the stability of neutral molecules~\cite{lieb1977hartreefock,catto1993binding}, the bulk limit for certain perfect crystals~\cite{catto1998mathematical, catto2001thermodynamic}, the existence of electronic ground states for crystals with local defects~\cite{cances2008new}, disordered crystals such as doped semiconductors~\cite{cances2013meanfield}, and standalone homogeneous 2D materials~\cite{gontier2021density}, and many other results.

    \medskip
    
    \noindent
    In this contribution, we focus on periodic or quasi-periodic 2D structures, for which an {\em rHF energy per unit area} can be defined. Examples of periodic 2D materials include monolayers, homo-bilayers (such as bilayer graphene in AA or AB stacking), homo-multilayers (such as rhombohedral graphene), and twisted homo-bilayers or homo-multilayers with commensurate twist angles. The geometric structure of these physical systems (with a focus on bilayers), for both the periodic and quasi-periodic settings, as well as the concept of stationary function, are recalled in Section~\ref{sec:geometric_structure}. 
    
    \medskip
    
    \noindent
    As in most DFT calculations performed in the encapsulated setting, the insulating encapsulating material is modeled as a continuous dielectric medium. In this approach, only the electrons of the 2D material under study (shown in red in Fig.~\ref{fig:encapsulated}) are explicitly taken into account.
    In Section~\ref{sec:electrostatics}, we provide useful results on electrostatics in the encapsulated setting. A key feature of this setting is that the conducting electrodes impose Dirichlet boundary conditions on the electrostatic potential, which effectively screen the Coulomb interaction.
    Most of the results presented are classical and follow from the theory of linear elliptic equations in divergence form, but we recall them for the sake of completeness. We also provide formulae for solving the Poisson equation in the special case where the dielectric permittivity only depends on the $z$ variable, which are useful for numerical simulations.
    
    \medskip
    
    Section~\ref{sec:quasi-periodic} deals with the quasi-periodic case. It covers twisted homo-bilayers, notably TBG, at generic (incommensurate) twist angles, as well as quasi-periodic multilayer structures such as generic twisted trilayer graphene (TTG) configurations. A large body of the mathematical literature is devoted to the study of linear quasi-periodic Hamiltonians~\cite{shubin1978almost, avila2009ten,simon1982almost,eliasson1992floquet,pastur1992spectra,li2026spectral}, see also~\cite{cances2017generalized,zhou2019plane, liu2022topological,etter2020modeling, wang2025convergence,cances2025numerical} for numerical methods.
     But to the authors' knowledge, very few contributions deal with nonlinear DFT-like models for quasi-periodic systems. The only works we are aware of are the studies of the quasi-periodic Thomas--Fermi (TF) and Thomas--Fermi--von Weizsäcker (TFW) models by Blanc, Le Bris and Lions~\cite{blanc2007energy}, which fall within the scope of orbital-free DFT. To handle quasi-periodic Kohn--Sham models, for which the main variable is the one-body density matrix (1-RDM), we rely on the $C^*$ and von Neumann algebra framework introduced by Bellissard and collaborators~\cite{bellissard2002coherent,bellissard1994noncommutative}, based on ideas of Connes~\cite{connes1994noncommutative}. We adapt it to the encapsulated setting and prove the existence of an rHF ground state for quasi-periodic systems in Section~\ref{sec:quasi-periodic}. The main arguments are inspired from those in~\cite{cances2013meanfield}.
    
    \medskip
    
    Finally, for completeness, the encapsulated Kohn--Sham model for periodic systems is detailed and analyzed in Section~\ref{sec:periodic}. The existence of an rHF ground state in the periodic setting is a simple exercise, and this result can be easily extended to some non-convex models such as Kohn--Sham LDA (local density approximation).

    \section{Geometric description of the physical systems under consideration}
    \label{sec:geometric_structure}
    
     We begin by describing in detail the geometry of the systems in the two settings covered by this analysis: the periodic one, and the quasi-periodic one. To simplify the presentation, we focus on bilayer systems, but our arguments obviously extend to generic periodic monolayer and multilayer systems, as well as to stationary ergodic multilayer systems.

    \subsection{Encapsulated 2D materials}
    
    As in implicit solvent models commonly used in quantum chemistry~\cite{tomasi2005quantum}, only the electrons of the 2D material are dealt with explicitly. This is justified by the fact that the Fermi level of the system is in the band gap of the encapsulating insulator and that we are only interested in low-energy phenomena which cannot create electron-hole pairs in the insulator. 
    
    \medskip
    
    This means that we will study a three dimensional system, but in which the electronic density is localized around the plane $z = 0$ (Fig.~\ref{fig:encapsulated}). We denote by $L > 0$ the (fixed) distance between the two insulating layers, and set
    \begin{equation*}
        I_L := \left(-\frac{L}{2}, \frac{L}{2}\right) \quad \text{and} \quad S_L := \R^2 \times I_L.
    \end{equation*}
    We often write $(x, z) \in \R^2 \times I_L$, with $x = (x_1, x_2) \in \R^2$ the two longitudinal coordinates, and $z$ the one orthogonal to the system.
    
    \medskip
    
    Let us point out that, since $I_L$ is bounded, electrons cannot escape to $\pm\infty$ in the $z$-direction. This is one of the key advantages of the encapsulated setting for the mathematical analysis. Unfortunately, most of our results do not extend to standalone 2D materials (see~\cite{gontier2021density, gontier2023density} for the case of homogeneous 2D materials in 3D space).

    \subsection{Bilayer 2D materials}
    
    Still for the sake of simplicity, we assume that each layer is a perfectly periodic monolayer but our results extend to the case when atomic relaxation is taken into account (see Remark~\ref{rem:atomic_relaxation}). In this simple case considered, the nuclear charge density $m$ is of the form
        \begin{equation*}
        m(x, z) = m_1(x, z) + m_2(x, z),
    \end{equation*}
    where, for all $j=1,2$ and $z \in I_L$, the map $x \mapsto m_j(x, z)$ is $\L_j$--periodic, with $\L_j$ a lattice of $\R^2$. We set 
        \begin{equation*}
        \Omega_1 := \R^2 / \L_1,\quad \Omega_2 := \R^2 / \L_2 \quad \text{and} \quad \Omega := \Omega_1 \times \Omega_2.
    \end{equation*}
    The sets $\Omega_1$ and $\Omega_2$ are two 2D-tori, and the set $\Omega$ is a 4D-torus called the configuration space.
    
    \medskip
    
    For a {\bf configuration} $\omega = (\omega_1, \omega_2) \in \Omega = \Omega_1 \times \Omega_2$, we set
    \begin{equation} \label{eq:atomic_charge_density}
        m_\omega(x, z) := m_1(x - \omega_1, z) + m_2 (x - \omega_2, z).
    \end{equation}
    In other words, $\omega_j$ encodes the shift by $\omega_j$ of the $j$-th lattice in the longitudinal direction. We introduce the group action $\R^2 \ni x \mapsto \alpha_x$ defined by 
    \begin{equation*}
        \alpha_x : \Omega \to \Omega, \qquad \alpha_x(\omega_1, \omega_2) := (\omega_1 - x, \omega_2 - x) \in \Omega,
    \end{equation*} 
     where $\omega_1 - x$ and $\omega_2 - x$ are seen as elements of the torus $\Omega_1$ and $\Omega_2$ respectively, that is are taken modulo $\L_1$ and $\L_2$ respectively. Note that $\alpha_x$ is measure preserving (for the Lebesgue measure for $\Omega$).
         
    \medskip
    
    In what follows, we consider functions $f = f(\omega, x, z) : \Omega \times \R^2 \times I_L \to \C$, and we often write $f_\omega(x, z)$ for $f(\omega, x, z)$. We say that such a function $f$ is {\bf stationary} if
    \begin{equation*}
        \forall x' \in \R^2, \qquad   f(\omega, x, z) = f(\alpha_{x'} \omega, x - x', z).
    \end{equation*}
    In particular, taking $x' = x$ leads to $f(\omega, x, z) = f(\alpha_x \omega, 0, z)$. In fact, stationary functions are of the form $f(\omega, x, z) := \ff(\alpha_x \omega, z) = \ff(\omega_1 - x, \omega_2 - x, z)$ for some $\ff : \Omega \times I_L \to \C$, called the {\bf generating function} of $f$. For instance, $m(\omega, x, z)$ is a stationary function with generating function
    \begin{equation*}
        m_\omega(x, z) := \fm(\alpha_x \omega, z), \quad \text{where} \quad \fm(\omega, z) := m_1(-\omega_1, z) + m_2 (-\omega_2, z).
    \end{equation*}
    
    \begin{remark}
    We restrain our study to bilayer 2D materials, but the case of multi-layers can be tackled the same way, with no additional mathematical difficulty.
    \end{remark}
 
     \subsection{Periodic vs incommensurate systems}
    
    In the modeling of bilayer materials, two cases are possible depending on the lattices $\L_1$ and $\L_2$.
    
    \medskip
    
    \underline{The periodic case  (commensurate lattices).} Consider first the case where $\L  := \L_1 \cap \L_2$ is also a discrete lattice of $\R^2$. In this case, for all $\omega \in \Omega$, the charge $m_\omega$ is a $\L $--periodic function. In this case, the physical properties of the system depend on the configuration $\omega \in \Omega$. For instance, in the case of two parallel graphene sheets with $\L_1 = \L_2 = \L =\Z a_1+\Z a_2$ representing the same triangular lattice, the AA stacking corresponds to $\omega_1=\omega_2$, while the AB stacking (the equilibrium configuration for standalone untwisted bilayer graphene) corresponds to $\omega_1=\omega_2+\frac 13 (a_1+a_2)$. These two configurations have very different electronic properties. 
    
    In the general case of twisted materials with commensurate twist angle, it may happen that the unit cell of $\L $ contains a large number of atoms, so that numerical computations can be very challenging. But from a mathematical viewpoint, we are in a periodic setting, which has been studied extensively in the literature~\cite{catto2001thermodynamic, catto2002periodic,gontier2016convergence}. The main difference with previous works comes from the encapsulated setting, which changes the electrostatic properties of the Coulomb interaction. We refer to Section~\ref{sec:periodic} below for our results in the periodic case, and additional comments on the structural difference with the quasi-periodic case.
    
    \medskip

    \underline{The quasi-periodic case (incommensurate lattices).} The group action $\alpha$ is called {\bf ergodic} on $\Omega$ (for the Lebesgue measure) if the only measurable functions $\ff:\Omega\to\C$ that are $\alpha$-invariant, that is $\ff\circ \alpha_x = \ff$ for all $x\in\R^2$, are the constant functions. In our case, Fourier decomposing $\ff : \Omega=\Omega_1 \times \Omega_2 \to \C$ as
    \begin{equation*}
    \ff(\omega) = \sum_{G\in\LL^*} \widehat \ff_G \re^{\ri (G_1 \omega_1 + G_2 \omega_2 )}, \quad \text{so that} \quad
    \ff(\alpha_x \omega) = \sum_{G\in\LL^*} \widehat{\ff}_G \re^{- \ri x (G_1 + G_2)} \re^{\ri (G_1 \omega_1 + G_2 \omega_2)} ,
    \end{equation*}
    (the sum runs over $G = (G_1, G_2) \in \L_1^* \times \L_2^*=:\LL^*$), we see that  $\ff\circ \alpha_x = \ff$ for all $x$ iff $\widehat \ff_G = 0$ for all $G$ with $G_1 + G_2 \neq 0$. We deduce that the action is ergodic if and only if the two lattices are {\bf incommensurate}, in the sense $\L_1^* \cap \L_2^* = \{ 0 \}$, that is
    \begin{equation*}
        \forall\ G_1 \in \L_1^*, \quad \forall\ G_2 \in \L_2^*, \qquad G_1 + G_2 = 0 \quad \Longleftrightarrow \quad G_1 = G_2 = 0.
    \end{equation*}
    
     The Birkhoff-Von Neumann ergodic theorem~\cite{moore2015ergodic,pastur1992spectra,carmona2012spectral} states that if $f(\omega, x) = \ff(\alpha_x \omega, 0)$ is stationary,  and if $\omega \mapsto f(\omega, x_0)$ is integrable for some $x_0 \in \R^2$, then we have, setting $B_R:=\{x\in\R^2:\|x\|<R\}$,
    \begin{equation}\label{eq:Birkhoff_thm}
        \lim_{R \to \infty}  \fint_{B_R} f(\omega_0, x) \rd x = \fint_{\Omega} f(\omega, x_0) \rd \omega = \fint_\Omega \ff(\omega) \, \rd \omega,
    \end{equation}
    for almost every $x_0 \in \R^2$ and almost every $\omega_0 \in \Omega$. Here, we use the notation $\fint_\cA := \frac{1}{| \cA | } \int_\cA$ for normalized integrals on the specified domain. In other words, the {\em averaged integral per unit area} is well-defined for stationary functions, and independent of the configuration $\omega_0 \in \Omega$.
    
    \medskip
    
    In view of the Birkhoff--Von Neumann ergodic theorem, we expect that, in the incommensurate case, the physical properties of the material are independent of the configuration $\omega \in \Omega$.
    
    \medskip

    Note that {\bf intermediate cases} are also possible. For instance when $\L_1$ and $\L_2$ are commensurate in one direction but incommensurate in the other (take $\L_1 = \Z \oplus \Z$ and $\L_2 = \Z \oplus \sqrt{2} \Z$). We do not consider such situations in the present article.
        
    \medskip

     \begin{remark}[Atomic relaxation] \label{rem:atomic_relaxation} As already mentioned in the introduction, atomic relaxation plays an important role in the physics of multilayer 2D materials. For unrelaxed periodic configurations, and in the absence of translation symmetry breaking, the relaxed configuration is also periodic and our arguments are unchanged. In the incommensurate case, the configuration space formalism, used in this contribution to study the electronic structure of the system, can be also used to study atomic relaxation in multilayer 2D materials~\cite{cazeaux2019approximation} (here also under a non-symmetry-breaking assumption on atomic displacements). After relaxation, the nuclear distribution is no longer the sum of two periodic functions, but still is a stationary function of the form $m_{\rm relaxed}(\omega,x,z)$. For these reasons, all our results also apply to relaxed configurations (in the absence of symmetry breaking). 
    \end{remark}

    \section{Encapsulated electrostatics}
    \label{sec:electrostatics}
    
    The encapsulated setting modifies the effective Coulomb interaction between the charged particles (nuclei and electrons) of the 2D material. As will be seen below, the encapsulated Coulomb potential shares properties which are closer to a Yukawa potential (exponential decay of the Green's kernel) than a full space Coulomb potential (polynomial decay of the Green's kernel). This will allow us to define the encapsulated Coulomb potential for densities which do not decay at infinity (such as periodic or stationary densities).

    \medskip
    
    Recall that the thin layers of insulating material are modeled by a polarizable continuum medium. We model this with dielectric permittivity $\epsilon$ depending solely on the transverse coordinate $z$ such that
    \begin{equation}\label{eq:bounds_epsilon}
        1 \le \epsilon(z) \le \epsilon_{\rm r} \quad \mbox{a.e.},
    \end{equation}
    where $\epsilon_{\rm r}$ is the dielectric permittivity of the insulating material. The potential difference $v_0$ between the lower and upper electrodes is tunable. Without loss of generality, we can assume that the potentials of the lower and upper electrodes are $v_0/2$ and $-v_0/2$ respectively (see Fig.~\ref{fig:encapsulated}).

    \subsection{Finite charge distributions}
    
    The results in this section are classical and are just recalled for the sake of completeness.
    Let $f \in H^{-1}(S_L)$ be an electronic charge density (in what follows, $f = \rho - m$ is the difference between the electronic density $\rho \ge 0$ and the charge density $m \ge 0$). The electrostatic potential $V_{f,v_0}$ generated by this charge density between the electrodes is solution to the inhomogeneous Dirichlet problem
    \begin{equation} \label{eq:VCoulomb}
        \begin{cases}
            - {\rm div}( \epsilon \nabla V_{f,v_0}) = 4 \pi f \quad \text{in} \quad S_L, \\
            V_{f,v_0}(\cdot, \pm L/2)  = \pm v_0/2.
        \end{cases}
    \end{equation}
    
    We will mainly focus on the case $v_0 = 0$, and study the Dirichlet problem
    \begin{equation} \label{eq:VCoulomb_bis}
        \text{Find} \quad V_f \in H^1_0(S_L) \quad \text{solution to} \quad  - {\rm div}( \epsilon \nabla V_f) = 4 \pi f \quad \text{in} \quad S_L.
    \end{equation}
    In this case, the boundary condition $V_f(x, \pm L/2) = 0$ has been encoded through the Sobolev space $H^1_0(S_L)$. The reduction to the case $v_0 = 0$ is justified by the following decomposition
    \begin{equation} \label{eq:Vv0-V0}
    V_{f, v_0} = V_{f} + v_0 g_0,
    \end{equation}
    where $g_0 = g_0(z) \in H^1(I_L)$ is the function
    \begin{equation} \label{eq:def:g0}
    g_0(z) :=  \frac{\int_{-L/2}^{z}\frac{\rd t}{\epsilon(t)}}{\int_{-L/2}^{L/2}\frac{\rd t}{\epsilon(t)}}-\frac{1}{2},
    \quad \text{which solves} \quad
    \begin{cases}
        - \partial_z (\epsilon \partial_z g_0) & = 0 \quad \text{in} \quad I_L \\
        g_0(\pm L/2) & = \pm 1/2
    \end{cases} . 
    \end{equation}
    
    Equation~\eqref{eq:VCoulomb_bis} is a classical problem. First, we recall that the weak formulation of~\eqref{eq:VCoulomb_bis} is
    \begin{equation} \label{eq:LaxMilgram}
        \forall h \in H^1_0(S_L), \qquad \int_{S_L} \epsilon \nabla V_f \cdot \nabla h = 4 \pi \langle f, h \rangle_{H^{-1}(S_L),H^1_0(S_L)}.
    \end{equation}
    Recall Lax--Milgram theory, which states the following.
    \begin{proposition}
        Let $f \in H^{-1}(S_L)$. There exists a unique $V_f \in H^1_0(S_L)$ solution to~\eqref{eq:VCoulomb_bis}. 
    \end{proposition}
    
    For brevity, we skip the elementary proof of this result, but emphasize that it relies on the Poincaré inequality in the $z$--direction, which implies that
    \begin{equation} \label{eq:Poincare}
        \forall V \in H^1_0(S_L), \quad \| V \|_{L^2(S_L)} \le \frac{L}{\pi} \| \nabla V \|_{L^2(S_L)}.
    \end{equation}
    In particular, the homogeneous norm $V \mapsto \| \nabla V \|_{L^2(S_L)}$ is equivalent to the $H^1(S_L)$ norm, and we have the embedding $H^1_0(S_L) \hookrightarrow L^2(S_L)$. We also have $H^1_0(S_L) \hookrightarrow L^6(S_L)$ by Sobolev inequality.
    
    For any $f \in H^{-1}(S_L)$, we define the {\bf encapsulated Hartree energy} of $f$ by
    \begin{equation*}
        \cD(f, f) := \langle f,V_f \rangle_{H^{-1}(S_L),H^1_0(S_L)}  =   \frac{1}{4 \pi} \int_{S_L} \epsilon | \nabla V_f |^2 \ge 0.
    \end{equation*}
    The last inequality, together with~\eqref{eq:Poincare}, shows that the quadratic form $f \mapsto \cD(f,f)$ is positive definite. In particular, it is strictly convex. We also recall a well-known result.
    \begin{lemma}
        If $f \ge 0$, then $V_f \ge 0$ as well.
    \end{lemma}

    \begin{proof}
        We take $h = V_{-} = \max \{ 0, -V_f \} \in H^1_0(S_L)$ in the weak-formulation~\eqref{eq:LaxMilgram}, and get
        \begin{equation*}
        0 \ge - \int_{S_L} \epsilon | \nabla V_- |^2 = 4 \pi \langle f, V_-\rangle_{H^{-1}(S_L),H^1_0(S_L)} \ge 0.
        \end{equation*}
        So both terms are null. Hence $\nabla V_- = 0$, so that, in view of~\eqref{eq:Poincare}, $V_- = 0$ almost everywhere. So $V_f \ge 0$.
    \end{proof}

        Actually, using Stampacchia's theory~\cite{stampacchia1965probleme} (see also~\cite[Chapter 8]{gilbarg2001elliptic}), we can prove the following
        \begin{lemma} \label{lem:Stampacchia}
             If $f \in H^{-1} (S_L) \cap L^p(S_L)$ with $p > \tfrac32$, then $V_f \in L^\infty(S_L)$, and there is a constant $C_p > 0$ such that
             \begin{equation*}
                \forall f \in H^{-1} (S_L) \cap L^p(S_L), \quad \| V_f \|_{L^\infty(S_L)} \le C_p \| f \|_{L^p(S_L)}^{\frac{p}{5p-6}} \| f \|_{H^{-1}(S_L)}^{\frac{4p-6}{5p-6}}.
             \end{equation*}
             In addition,  for all $\tfrac{3}{2} < p \le 2$, there is a constant $C_{p,L} > 0$ such that
             \begin{equation} \label{eq:continuity_Coulomb_Lp_Linfty}
                 \forall f \in L^p(S_L), \quad \| V_f \|_{L^\infty(S_L)} \le C_{p,L} \| f \|_{L^p(S_L)} \qquad \mbox{(case $\tfrac32 < p \le 2$).}
             \end{equation}
        \end{lemma}
        
        Let us provide a proof for completeness, which only uses the weak-formulation~\eqref{eq:LaxMilgram}, and works in particular in the case where $\epsilon(\cdot)$ is not a smooth function.
        
        \begin{proof}
            For $t > 0$, we define 
            \begin{equation*}  
                \phi_t := (V_f - t)_+ \in H^1_0(S_L), \quad A(t) := \{ V_f > t \} \subset S_L \quad \text{and} \quad \mu(t) := |A(t)|.
            \end{equation*}
            Taking $\phi_t$ as a test function in~\eqref{eq:LaxMilgram} using $\epsilon \ge 1$, and Hölder's inequality, we get
            \begin{equation}\label{eq:nabla_phit}
                \|\nabla \phi_t\|_{L^2(S_L)}^2 \le \int_{S_L} \epsilon |\nabla \phi_t|^2 = 4\pi \int_{S_L} f \phi_t \le 4 \pi  \|f\|_{L^p(S_L)} \|\phi_t\|_{L^{p'}(S_L)}.
            \end{equation}
            Since $p > \frac32$, we have $p' < 3$ (and in particular, $1 \le p' \le 6$).
            On the other hand, using that $\phi_t \in H^1_0(A(t))$ with $A(t)$ a bounded set, we can use Hölder's inequality and the Sobolev inequality $H^1_0(S_L) \hookrightarrow L^6(S_L)$ so get
            \begin{equation*}
                 \|\phi_t\|_{L^{p'}(S_L)} \le \| \phi_t \|_{L^6(S_L)} | A(t) |^{\alpha} \le C \| \nabla \phi_t \|_{L^2(S_L)} \mu(t)^{\alpha} \quad \text{with} \quad  \frac{1}{6} + \frac{\alpha}{1} =  \frac{1}{p'} = 1 - \frac{1}{p}.
            \end{equation*}
            Injecting into~\eqref{eq:nabla_phit} gives
            \begin{equation}\label{eq:bound_nabla_phit}
                \|\nabla \phi_t\|_{L^2(S_L)} \le C\|f\|_{L^p(S_L)}\, \mu(t)^{\alpha}.
            \end{equation}
            Finally, for $s > t$, since $V_f - t \ge s - t$ on $A(s) \subset A(t)$, we also have the Markov type inequality
            \begin{equation*}
            (s-t)\,\mu(s) \le \int_{A(s)} (V_f - t) \le \|\phi_t\|_{L^6(S_L)}\,\mu(s)^{5/6} \le C\|\nabla\phi_t\|_{L^2(S_L)}\,\mu(s)^{5/6}.
            \end{equation*}
            Combined with~\eqref{eq:bound_nabla_phit}, this gives 
            \begin{equation}\label{eq:bound_mut}
                \mu(s) \le \frac{C\|f\|_{L^p(S_L)}^6}{(s-t)^6}\,\mu(t)^{\beta}, \qquad \beta := 6 \alpha = 5 - \frac{6}{p}.
            \end{equation}
            Since $p > 3/2$, we have $\beta > 1$. 
            
            \medskip
            
            In the original works by Stampacchia, which was restricted to bounded domains $B$, the iteration can be started at $t_0 = 0$ (using $\mu(t_0) \le | B |$). In our case however, $S_L$ is unbounded, and we need to start at $t_0 > 0$. To control $\mu(t_0)$, we use that $V_f \in H^1_0(S_L) \hookrightarrow L^6(S_L)$ with $\|V_f\|_{L^6(S_L)} \le C\|f\|_{H^{-1}(S_L)}$, and get the Markov inequality
            \begin{equation*}
            \mu(t_0) \le \frac{\|V_f\|_{L^6(S_L)}^6}{t_0^6} \le \frac{C\|f\|_{H^{-1}(S_L)}^6}{t_0^6} < \infty.
            \end{equation*}
            The Stampacchia iteration lemma~\cite{stampacchia1965probleme, kinderlehrer2000introduction} applied to~\eqref{eq:bound_mut} yields $\mu(t_0 + T) = 0$ with $T = C\|f\|_{L^p} \mu(t_0)^{(\beta-1)/6}$. So, for all $t_0 > 0$, we have
            \begin{equation*}
                \|V_f\|_{L^\infty(S_L)} \le t_0 + T \le t_0 + C\|f\|_{L^p}\mu(t_0)^{\tfrac{\beta-1}{6}} \le t_0 + C' \| f \|_{L^p} \|f\|_{H^{-1}}^{\beta-1} \frac{1}{t_0^{\beta - 1}}.
            \end{equation*}
            Optimizing over $t_0 > 0$ and taking $t_0 = \left( (\beta - 1) C' \|f\|_{L^p}\|f\|_{H^{-1}}^{\beta-1}\right)^{1/\beta}$ yields, as wanted,
            \begin{equation*}
            \|V_f\|_{L^\infty(S_L)} \le C'' \|f\|_{L^p(S_L)}^{\frac{1}{\beta}}\|f\|_{H^{-1}(S_L)}^{\frac{\beta-1}{\beta}} = C'' \|f\|_{L^p(S_L)}^{\frac{p}{5p-6}}\|f\|_{H^{-1}(S_L)}^{\frac{4p-6}{5p-6}}.
            \end{equation*}
            Applying the same argument to $-V_f$ concludes the proof of the first inequality of Lemma~\ref{lem:Stampacchia}. To prove the second one, we recall the Sobolev embedding $H^1_0(S_L) \hookrightarrow L^2(S_L) \cap L^6(S_L)$, see~\eqref{eq:Poincare}, so that $L^{2}(S_L) + L^{6/5}(S_L) \hookrightarrow H^{-1}(S_L)$, and in particular, $L^p(S_L) \hookrightarrow H^{-1}(S_L)$ for all $\tfrac32 < p \le 2$.

        \end{proof}

    For encapsulated 2D materials, the electronic density is not in $H^{-1}(S_L)$, since it does not decay in the longitudinal directions. In the physical settings considered here, we expect it to be either periodic or quasi-periodic. In order to prove that the corresponding potential is a bounded function, we need to study the Green's function $\mathcal{G}$ of the operator $-{\rm div}(\epsilon \nabla \cdot)$ in $S_L$ with homogeneous Dirichlet boundary conditions.

    \subsection{Green's function} 
    \label{ssec:Greenfunction} We denote by $\cL$ the self-adjoint operator 
    \begin{equation*}
        \cL u := - {\rm div}(\epsilon \nabla u), 
            \end{equation*} 
    on $L^2(S_L)$ with domain $H^2(S_L) \cap H^1_0(S_L)$ and form domain $H^1_0(S_L)$. Since $\epsilon$ only depends on the $z$-variable, we have
    \begin{equation} \label{eq:L_and_tildeL}
        \cL := - \Delta_x \otimes \epsilon(z) + \bbI_{L^2(\R^2)} \otimes \widetilde{\cL},
    \end{equation}
    where $\widetilde{\cL}$ is the self-adjoint operator $L^2(I_L)$ with domain $H^2(I_L) \cap H^1_0(I_L)$ and form domain $H^1_0(I_L)$, defined by $\widetilde{\cL}(u) := - \partial_z (\epsilon \partial_z u)$. The operator $\widetilde{\cL}$ is coercive and has a compact resolvent. 
    
    \begin{lemma} \label{lem:spectral_gap}
        We have the operator bound $\cL \ge \frac{\pi^2}{L^2}> 0$. 
    \end{lemma}
        
    \begin{proof}
        Let $\mu_1$ denote the lowest eigenvalue of $\widetilde{\cL}$. Since $\epsilon \ge 1$, we have $\widetilde{\cL} \ge - \partial_{zz}^2$ on $H^1_0(I_L)$. Considering the lowest eigenvalue of both operators, we get $\mu_1 \ge \frac{\pi^2}{L^2}$. Finally, we have $\cL \ge \bbI \otimes \widetilde{\cL}$, so $\cL \ge \mu_1 \ge  \frac{\pi^2}{L^2}$ as well.
    \end{proof}

    The fact that $\cL \ge \mu_1 > 0$, hence has a spectral gap at energy $0$, shows that the encapsulated Coulomb potential shares properties which are closer to those of a Yukawa potential than to the usual 3D Coulomb potential. In some sense, the Dirichlet boundary conditions screen the electrostatic potential. 
     
     \medskip
     
    Let us define the Green's function $\cG(x, z ; x',  z') := \cG_{x', z'}(x, z)$ as the distributional solution to
    \begin{equation*}
    - {\rm div}( \epsilon \nabla \cG_{x', z'}) = 4 \pi \delta_{x', z'}, \qquad \cG_{x', z'}(x, \pm L/2) = 0.
    \end{equation*}
    This function can be seen as the kernel of the operator $\cL^{-1}$ from $L^2(S_L)$ to itself. By longitudinal-translation invariance, we have $\cG(x, z ; x', z') = \cG(x - x' ; z, z')$. By the maximum principle, the function $\mathcal G$ is nonnegative.

\subsection{Explicit expression of the Green's function}

    The function $\cG$ is the kernel of the operator $\cL^{-1}$ with
    \begin{equation*}
        \cL := - \Delta_x \otimes \epsilon(z) + \bbI_{L^2(\R^2)} \otimes \widetilde{\cL},
    \end{equation*}
    and where $\widetilde{\cL}$ is the self-adjoint operator on $L^2(I_L)$ with form domain $H^1_0(I_L)$ defined by $\widetilde{\cL}(u) := - \partial_z (\epsilon \partial_z u)$. Because of the $\epsilon$ factor in the first term, we consider the modified 1D operator
    \begin{equation*}
    \widetilde{\cL}_\epsilon u := - \frac{1}{\epsilon} \partial_z (\epsilon \partial_z u),
    \end{equation*}
    acting on the weighted Hilbert space $L^2_\epsilon := L^2(I_L, \epsilon(z) \rd z)$, with form domain $H^1_0(I_L)$. We denote by
    \begin{equation*}
    \langle f, g \rangle_{L^2_\epsilon} := \int_{I_L} \overline{f(z)} g(z) \epsilon(z) \rd z.
    \end{equation*}
    It can be checked that $\widetilde{\cL}_\epsilon$, with form domain $H^1_0(I_L)$, is self-adjoint on $L^2_\epsilon$, positive and with compact resolvent. We consider its spectral decomposition of the form
    \begin{equation*}
    \widetilde{\cL}_\epsilon = \sum_{k=1}^\infty \lambda_k | g_k \rangle \langle g_k |,
    \qquad 0 < \lambda_{1} \le \lambda_{2} \le \cdots, 
    \qquad  \langle g_{k}, g_{\ell} \rangle_{L^2_{\epsilon}} := \int_{I_L} \epsilon g_{k} g_{\ell} = \delta_{k \ell}.
    \end{equation*}
    In particular, the eigenstates $g_k$ are solutions to the generalized elliptic eigenproblem
    \begin{equation*} 
        - \partial_z \left(\epsilon \partial_z g_{k} \right) =  \lambda_{k} \epsilon g_{k}, \qquad g_k \in H^1_0(I_L).
    \end{equation*} 
    
    \medskip
    
    Consider a charge distribution $f \in C^\infty_0(S_L)$ and the corresponding potential $V_f$, so that
    \begin{equation} \label{eq:eqt_V_f_expansion}
        \epsilon(z)  (- \Delta_x V_f(x,z))  - \partial_z \big( \epsilon(z) \partial_z V_f (x,z) \big) = 4 \pi f(x,z) = 4 \pi \epsilon(z) \left( \epsilon(z)^{-1} f(x, z) \right).
    \end{equation}
    For each $x \in \R^2$, we expand $\epsilon^{-1} f(x, \cdot)$ and $V_f(x, \cdot)$ in the $(g_{k})$ basis (for the $L^2_\epsilon$ inner product), giving an expansion of the form
    \begin{equation} \label{eq:expansion_rho_V}
        f(x,z) = \epsilon(z) \sum_{k=1}^\infty f_k(x) g_{k}(z), \qquad V_f(x, z) = \sum_{k=1}^\infty V_{k}(x) g_{k}(z).
    \end{equation}
    We multiply~\eqref{eq:eqt_V_f_expansion} by $g_k(z)$ and integrate on $I_L$. Using the orthonormality of the $(g_k)$ basis for the $L^2_\epsilon$ inner product, we get
    \begin{equation} \label{eq:PBn}
        -\Delta_x V_k(x) + \lambda_k V_k(x) = 4 \pi f_k(x),
    \end{equation}
    from which we infer that
    \begin{equation*}
    V_k =  2 f_k \star Y_{\sqrt{\lambda_k}},
    \end{equation*}
    where $Y_\kappa$ denotes the 2D Yukawa potential, that is the Green's function of the Poisson equation
    $$
    -\Delta Y_\kappa + \kappa^2 Y_\kappa = 2\pi \delta_0 \quad \mbox{in } \mathcal S'(\R^2).
    $$
    Recall that we have
    $$
    \forall x \in \R^2 \setminus \{0\}, \quad Y_\kappa(x) = Y_1(\kappa x), \quad Y_1(x) = K_0(|x|),
    $$
    where $K_0$ is the modified Bessel function of the second kind. We finally get
    \begin{equation} \label{eq:explicit_formula_Green}
        \cG(x,z,z') = 2 \sum_{k=1}^{\infty} Y_{\sqrt{\lambda_k}}(x)  g_k(z)  g_k(z') = 2 \sum_{k=1}^{\infty} K_0\left({\sqrt{\lambda_k}} |x|\right)   g_k(z)  g_k(z') .
    \end{equation}
    
    This formula allows to prove the exponential decay of the Green's kernel.
    
    \begin{lemma} \label{lem:CombesThomas}
        For all compact $\Gamma \subset \R^2$ containing a neighborhood of $0$, there exist $C_\Gamma > 0$ and $\mu > 0$ so that
        \begin{align*}
            \forall (x, z), (x',z') \in S_L, \qquad (x - x' ) \notin \Gamma \implies  | \cG(x, z;x', z') |  \le C_\Gamma \, \re^{ - \sqrt{\mu} | x - x' |}.
        \end{align*}
    \end{lemma}

    \begin{proof}  
    The functions $g_k$ are in $H^1_0(I_L)$, are $L^2_\epsilon$-normalized, i.e. $\int_{I_L} \varepsilon(z) | g_k |^2(z) \, dz= 1$, and satisfy
    \begin{equation} \label{eq:GEP}
        - \partial_z \left(\epsilon \partial_z g_{k} \right) =  \lambda_{k} \epsilon g_{k}.
    \end{equation}
    It is a classical result that $\lambda_k \mathop{\sim}_{k \to \infty} \frac{\pi^2 k^2}{L^2}$~\cite[Thm. 4.3.1]{zettl2005sturm}. Multiplying~\eqref{eq:GEP} by $g_k$ and integrating shows that
    \begin{equation*}
        \| g_k' \|_{L^2(I_L)}^2 \le \int_{I_L} \epsilon | g_k' |^2 = \lambda_k.
    \end{equation*}
    A variant of the Gagliardo--Nirenberg inequalities then shows that $g_k \in L^\infty(I_L)$ with $\| g_k \|_\infty \le \sqrt2\,\lambda_k^{1/4}$. Indeed, we have, for any $z \in I_L$, since $\epsilon\ge 1$,
    \begin{equation*}
        g_k^2(z) = \int_{-\frac{L}{2}}^z 2 g_k(s) g_k'(s) \rd s \le 2 \| g_k \|_{L^2(I_L)} \| g_k' \|_{L^2(I_L)} \le 2 \sqrt{\lambda_k}.
    \end{equation*}
       
    \medskip
    
    Recall that the Bessel function  $K_0$ is a smooth positive function on $(0,+\infty)$, satisfying
    \begin{equation*}
        K_0(s) \mathop{\sim}_{s \to 0^+}  - \ln(s), \qquad K_0(s) \mathop{\sim}_{s \to +\infty} \sqrt{\frac{\pi}{2s}} \re^{-s}.
    \end{equation*}
    In particular, for all $s_0 > 0$, there is $C \ge 0$ such that
    \begin{equation*}
        \forall s \ge s_0, \quad 0 < K_0(s) \le C \re^{-s}, \quad |K_0'(s)| \le C \re^{-s}.
    \end{equation*}
    Together with the formula~\eqref{eq:explicit_formula_Green} for the Green's function and the bound on $g_k$, we obtain that, for all $|x| \ge R_0:=s_0/\sqrt{\lambda_1}$, 
    \begin{equation*}
        \left| \cG(x,z,z') \right| \le C \sum_{k=1}^\infty \| g_k \|_\infty^2 \re^{- \sqrt{\lambda_k} | x | } 
        = C \left( \sum_{k=1}^\infty \sqrt{\lambda_k} \re^{-(\sqrt{\lambda_k} - \sqrt{\lambda_1}) |x|} \right) \re^{- \sqrt{ \lambda_1} | x |}.
    \end{equation*}
    The last series converges in view of the asymptotics $\lambda_k \mathop{\sim}_{k \to \infty} \frac{\pi^2 k^2}{L^2}$. This proves Lemma~\ref{lem:CombesThomas} with $\mu = \lambda_1$.
    \end{proof}

   \subsection{Charge distributions in uniform Lebesgue spaces.}
    \label{ssec:Coulomb_Lpunif}
    As we already mentioned, the electronic density of an encapsulated 2D material is not in $H^{-1}(S_L)$, so we also need to extend the definition of $V_f$ to larger classes of charge distributions $f$. For instance, in the periodic setting, the electronic density is a locally-integrable $\L$-periodic function on $S_L$. For completeness, let us work in the larger $L^p_{\rm unif}(S_L)$ class, defined for $1 \le p \le \infty$ by
    \begin{align*}
        L^p_\unif(S_L) &:= \left\{ f \in L^p_\loc(S_L), \quad  \|f\|_{L^p_\unif(S_L)}:= \sup_{x\in\R^2} \| f \|_{L^p((x+[0,1]^2)\times I_L)} <\infty \right\}.
    \end{align*}
    Note that $L^\infty_{\rm unif}(S_L)=L^\infty(S_L)$, that $L^q_{\rm unif} \subset L^p_{\rm unif}$ if $p \le q$ (this comes from Hölder's inequality), and that $L^p_{\rm loc}$ functions which are $\L$-periodic with respect to the longitudinal variable $x$ are in $L^p_{\unif}(S_L)$ as well.
    
    \begin{lemma} \label{lem:potential_Lp_unif}
        If $f \in L^p_{\rm unif}(S_L)$ for some $p > \tfrac{3}{2}$, then for all $(x,z) \in S_L$, the integral
        \begin{equation*}
            V_f(x,z) := \int_{S_L} \cG(x, z ; x', z') f(x', z') \rd x' \rd z'
        \end{equation*}
        is well-defined. The function $V_f$ satisfies the Dirichlet problem~\eqref{eq:VCoulomb} with $v_0 = 0$. In addition, $V_f \in L^\infty(S_L)$, and there is a constant $C_{p,L}^{\rm unif}>0$ such that 
        \begin{equation*}
        \forall f \in L^p_{\rm unif}(S_L), \qquad \| V_f\|_{L^\infty(S_L)}  \le C_{p,L}^{\rm unif} \| f \|_{L^p_{\rm unif}(S_L)}.
        \end{equation*}
    \end{lemma}
    In particular, $V_f$ belongs to all $L^q_{\rm unif}(S_L)$ spaces for all $1 \le q \le \infty$.
    
    \begin{proof}
        By longitudinal-translation invariance of the problem, it is enough to bound $V_f(x_0, z_0)$ for $x_0 = 0$. Consider a compact subset $\Gamma$ of $\R^2$ containing a neighbourhood of $x_0 = 0$. We write $f = f_1 + f_2$ with
        \begin{equation*}
        f_1 (x, z) = f(x, z) \1(x \in \Gamma), \qquad f_2(x, z) = f(x, z) \1(x \notin \Gamma).
        \end{equation*}
        We have similarly $V_f = V_{f_1} + V_{f_2}$. We bound the two terms independently.
        Since $f_1 \in L^p(S_L)$, Lemma~\ref{lem:Stampacchia} shows that $V_{f_1} \in L^\infty(S_L)$, with
        \begin{equation*}
           \left\| V_{f_1} \right\|_{L^\infty(S_L)} \le C_p \| f_1 \|_{L^p(S_L)}^{\frac{p}{5p-6}} \| f_1 \|_{H^{-1}(S_L)}^{\frac{4p-6}{5p-6}} \le C_{p, \Gamma, L} \| f_1 \|_{L^p(S_L)} \le \widetilde C_{p, \Gamma, L} \| f_1 \|_{L^p_{\rm unif}(S_L)}.
        \end{equation*}
        The last inequality comes from the fact that $f_1$ has compact support in $\Gamma \times I_L$, so that $L^p(\Gamma \times I_L) \hookrightarrow L^{6/5}(S_L) \hookrightarrow H^{-1}(S_L)$.
        
        \medskip
        
        To bound $V_{f_2}(0, z_0)$, we use the exponential decay of $\cG$ in Lemma~\ref{lem:CombesThomas}. We have
        \begin{equation*}
        | V_{f_2}(0, z_0) | \le C_\Gamma \int_{S_L} \re^{- \sqrt{\mu} |x' |} | f(x', z') | \rd x' \rd z'.
        \end{equation*}
        We decompose $\R^2$ as $\R^2= \bigcup_{k \in \Z^2} Q_k$ with $Q_k = k + [0, 1]^2$, and write 
        \begin{equation*}
        | V_{f_2}(0, z_0) | \le C' \sum_{k \in \Z^2}  \re^{- \sqrt{\mu} |k |} \int_{Q_k \times I_L} | f|  \le C' \| f \|_{L^1_{\rm unif}} \sum_{k \in \Z^2}  \re^{- \sqrt{\mu} | k |}.
        \end{equation*}
        The last sum is convergent. Finally, Hölder's inequality shows that
        $$\| f \|_{L^1_{\rm unif}(S_L)} \le L^{1/p'} \| f \|_{L^p_{\rm unif}(S_L)},$$
        and the result follows.
    \end{proof}
    
    In the case where $f$ is a $\L $--periodic function in the $x$ variable, then so is $V_f$. In this case, we can define the {\bf periodic Hartree energy per unit cell} as
    \begin{equation} \label{eq:def:Hartree_Periodic}
        \cD_{\rm per}(f,f) := \int_{(\R^2 / \L ) \times I_L} V_f(x,z) f(x,z) \rd x \rd z
         = \frac{1}{4 \pi} \int_{(\R^2 / \L ) \times I_L} \epsilon(z) | \nabla V_f |^2 \rd x \rd z \quad \ge 0.
    \end{equation}
    In the case where $f \in L^p_{\rm per}(S_L)$ for some $p > \frac32$, we have $V_f \in L^\infty(S_L)$, and the first integral is well defined (recall that $L^p_{\rm per}(S_L) \hookrightarrow L^1_{\rm per}(S_L)$).

    \subsection{Stationary charge distributions}
    \label{ssec:Coulomb_stationary}
    We now focus on the stationary case. Recall that $f_\omega(x,z)$ is stationary if it is of the form $\ff(\alpha_x \omega, z)$ for some generating function $\ff : \Omega \times I_L \to \C$, that is
    \begin{equation*}
        f_\omega(x, z) = \ff(\alpha_x \omega, z) = \ff(\omega_1 - x, \omega_2 - x, z).
    \end{equation*}
    If $\ff$ is a continuous function (hence bounded) in $C^0(\Omega \times I_L)$, then so is $f_\omega$. In this case, $f_\omega \in L^\infty(S_L)$, and one can define the potential $V_{f_\omega}$ as in Lemma~\ref{lem:potential_Lp_unif}, which is bounded. By translation invariance of $\cG$ along the two first dimensions, $V_f$ is also stationary, of the form 
    \begin{equation*}
    V_f(\omega, x, z) := \cV_\ff(\alpha_x \omega, z).
    \end{equation*}
          In the sequel, we work with generating functions $\ff$ in the Lebesgue space $L^p(\Omega \times I_L)$, which are not necessarily continuous. In this case, $f_\omega(x, z) := \ff(\alpha_x \omega, z)$ may not belong to any $L^p_{\rm unif}(S_L)$ space, and it is unclear at this point that the corresponding potential $V_f$ is still well-defined. The spaces $L^{r,p}$ introduced below are precisely designed to handle this difficulty
    
    \medskip
    
    For $1 \le  r,p < \infty$, we introduce the Lebesgue spaces
    $$
    L^{r,p}:=\left\{ \nu \in L^1(\Omega \times I_L), \quad
    \|\nu\|_{L^{r,p}}^r:=\int_{\Omega_2} \left( \int_{\Omega_1 \times I_L} |\nu(\omega_1,\omega_1+\omega_2,z)|^p \,\rd\omega_1 \,\rd z \right)^{r/p} \,\rd\omega_2 < \infty \right\}, 
    $$
    and analogous definitions for $r=\infty$ or $p=\infty$. After a change of variable $\omega_2' := \omega_1 + \omega_2$, this is also the mixed Lebesgue space
    \begin{equation*}
    L^{r,p} \simeq L^r(\Omega_2, L^p(\Omega_1 \times I_L)).
    \end{equation*}
    When $r = p$, we have $L^{p,p} = L^p(\Omega \times I_L)$. In addition, since $\Omega_1 \times I_L$ is bounded, Hölder's inequality shows that $L^{r,q}  \subset L^{r, p}$ if $p \le q$, and since $\Omega_2$ is compact, we also have $L^{s, p} \subset L^{r,p}$ if $r \le s$. In particular, $L^{r,p} \subset L^{\min(r,p)}(\Omega \times I_L) \subset L^1(\Omega \times I_L)$.
    
    \medskip
    
    The key result of this section is the following, which is the stationary equivalent to Lemma~\ref{lem:potential_Lp_unif}.
    
    \begin{theorem} \label{th:potential_Lrp}
        If $\ff \in L^{r,p}$ with $1 \le r \le \infty$ and $\tfrac32 < p \le \infty$ and $f_\omega(x,z):=\ff(\alpha_x\omega,z)$. Then for all $(x,z) \in S_L$ and almost all $\omega \in \Omega$, the integral
        \begin{equation*}
            V_{f_\omega} (x,z) := \int_{S_L} \cG(x, z ; x', z') f_\omega(x', z') \rd x' \rd z'
        \end{equation*}
        is well-defined. The function $V_{f_\omega}$ is solution to the Dirichlet problem~\eqref{eq:VCoulomb} for $f=f_\omega$ and $v_0 = 0$. The function $V_f(\omega, x, z) := V_{f_\omega}(x,z)$ is stationary, of the form $V_f(\omega, x, z) = \cV_\ff(\alpha_x \omega, z)$. In addition, $\cV_\ff \in L^{r,\infty}$, and there exists a constant $C_{r,p, L}^{\rm stat} > 0$ such that 
        \begin{equation*}
            \forall \ff \in L^{r,p}, \qquad \| \cV_\ff \|_{L^{r, \infty}} \le C_{r,p, L}^{\rm stat} \| \ff \|_{L^{r,p}}.
        \end{equation*}
    \end{theorem}
    
    \begin{proof}
        By linearity, it is enough to prove the result for $\ff$ a nonnegative measurable function in $L^{r,p}$. In this case, $f_\omega$ and $V_{f_\omega}$ are well-defined for almost $\omega \in \Omega$ and nonnegative since $\mathcal G \ge 0$ by the maximum principle. 
        
        It is useful to work with a configuration of the form $(\omega_1, \omega_1 + \omega_2)$. We have 
        \begin{align*}
            \cV_\ff(\omega_1, \omega_1 + \omega_2, z) & := V_{f_{\omega_1, \omega_1 + \omega_2}}(x = 0, z) 
            = \int_{S_L} \cG(- x', z, z') f_{\omega_1, \omega_1 + \omega_2}(x', z') \rd x' \rd z'.
        \end{align*}
        As in the proof of Lemma~\ref{lem:potential_Lp_unif}, we write
        \begin{equation*}
        f_{\omega_1, \omega_1 + \omega_2} = f_{\omega_1, \omega_1 + \omega_2}^{(1)} + f_{\omega_1, \omega_1 + \omega_2}^{(2)},
        \end{equation*}
        with 
        \begin{equation*}
        f_{\omega_1, \omega_1 + \omega_2}^{(1)}(x, z) := f_{\omega_1, \omega_1 + \omega_2}(x,z)  \1 (x \in \Gamma), \quad \text{and} \quad
        f_{\omega_1, \omega_1 + \omega_2}^{(2)}(x,z) := f_{\omega_1, \omega_1 + \omega_2}(x,z) \1 (x \notin \Gamma).
        \end{equation*}
        Here $\Gamma$ is a compact subset of  $\R^2$ containing a neighborhood of $0$, that will be specified later.
        We denote by $V^{(1)}$ and $V^{(2)}$ the corresponding potentials. To control $V^{(1)}$, we use that $f_{\omega_1, \omega_1 + \omega_2}^{(1)}$ is compactly-supported. As in the proof of Lemma~\ref{lem:potential_Lp_unif}, we get
        \begin{equation*}
         \| V^{(1)} \|_{L^\infty(S_L)}^p \le C \| f^{(1)} \|_{L^p(S_L)}^p
        = \int_{S_L} \1 (x' \in \Gamma) \ff^p(\omega_1 - x', \omega_1 + \omega_2 - x', z') \rd x' \rd z'.
        \end{equation*}
        Taking $\Gamma$ a shifted version of $\Omega_1$ containing a neighbourhood of $0$, we can make the change of variable $\omega_1' := - x' + \omega_1$, and get
        \begin{equation*}
        \left|  \cV_\ff^{(1)}(\omega_1, \omega_1 + \omega_2, z)  \right|^p 
        = \left|  V^{(1)}(0, z)  \right|^p  \le C \int_{\Omega_1} \ff^p(\omega_1', \omega_1' + \omega_2, z') \rd \omega_1' \rd z'.
        \end{equation*}
        The right-hand side no longer depends on $\omega_1$, so we can take the supremum in $\omega_1$ on the left-hand side. Integrating first in $z$ (and using that $| I_L | = L$), raising to the power $r/p$, and integrating in $\omega_2$ gives
        \begin{equation*}
        \left|  \cV_\ff^{(1)}  \right|_{L^{r, \infty}} \le C L \| \ff \|_{L^{r, p}}.
        \end{equation*}
        
        We now focus on the term $V^{(2)}$. This time, we use the exponential decay of $\cG$ in Lemma~\ref{lem:CombesThomas}, and get that there is a constant $C_\Gamma > 0$ such that
        \begin{equation*}
        | V^{(2)}(0, z) | \le C \int_{S_L} \re^{- \sqrt \mu | x' |} f_{\omega_1, \omega_1 + \omega_2}(x', z') \rd x' \rd z' 
        = C \int_{S_L} \re^{- \sqrt\mu | x' |} \ff(\omega_1 - x', \omega_1 + \omega_2- x', z') \rd x' \rd z' .
        \end{equation*}
        We write $x' = \omega_1' + \ell$ with $\omega_1' \in \Omega_1$ and $\ell \in \L_1$ and get
        \begin{equation*}
        | V^{(2)}(0, z) | \le C \sum_{\ell \in \L_1} \int_{\Omega_1 \times I_L} \re^{- \sqrt\mu |  \omega_1' + \ell |} \ff(\omega_1 - \omega_1', \omega_1 + \omega_2- \omega_1' - \ell, z') \rd \omega_1' \rd z'.
        \end{equation*}
        Using that $\re^{- \sqrt\mu | \omega_1' + \ell |} \le \re^{- \sqrt\mu | \ell |} \re^{ \sqrt\mu | \omega_1' |} \le C \re^{- \sqrt\mu | \ell |}$, and making the change of variable $\omega_1' = - \omega_1'' + \omega_1$ gives
        \begin{equation*}
        \left| \cV_\ff^{(2)} (\omega_1, \omega_1 + \omega_2, z) \right| =: \left|  V^{(2)}(0, z)  \right|  \le C' \sum_{\ell \in \L_1} \re^{- \sqrt\mu | \ell |} \underbrace{\int_{\Omega_1 \times I_L}  \ff(\omega_1'', \omega_1'' + \omega_2 - \ell, z') \rd \omega_1'' \rd z'}_{f_\ell(\omega_2)}.
        \end{equation*}
        Again, the right-hand side no longer depends on $\omega_1$, and we can take the supremum in $\omega_1$ on the left-hand side. On the other hand, Jensen's inequality with the convex function $t \mapsto t^r$ shows that
        \begin{equation*}
        \left( \sum_{\ell \in \L_1} w_\ell f_\ell  \right)^r \le W^{r-1} \left(\sum_{\ell \in \L_1} w_\ell f_\ell^r \right), \qquad W := \sum_{\ell \in \L_1} w_\ell.
        \end{equation*}
        With $w_\ell := \re^{- \sqrt\mu | \ell |}$ and $f_\ell = f_\ell(\omega_2)$, we deduce that
        \begin{equation*}
        \sup_{\omega_1 \in \Omega_1}\left| \cV_\ff^{(2)} (\omega_1, \omega_1 + \omega_2, z) \right|^r \le (C')^r W^{r-1} \sum_{\ell \in \L_1} \re^{- \sqrt\mu | \ell | } \left( \int_{\Omega_1 \times I_L}  \ff(\omega_1', \omega_1' + \omega_2 - \ell, z') \rd \omega_1' \rd z' \right)^r.
        \end{equation*}
        We integrate in $\omega_2 \in \Omega_2$ (the integral no longer depends on $\ell$), take the supremum in $z$, and get
        \begin{equation*}
        \left\| V_\ff^{(2)} \right\|_{L^{r, \infty}} \le C' W \| \ff \|_{L^{r, 1}} \le C'' \| \ff \|_{L^{r, p}},
        \end{equation*}
        where we used that $L^{r, p} \subset L^{r, 1}$ with continuous embedding by Hölder's inequality. The proof follows.
    \end{proof}

    We are now in position to define the {\bf stationary Hartree energy per unit area}. Denoting by $(e_{G_1,G_2})_{(G_1,G_2) \in \L_1^*\times \L_2^*}$ the orthonormal Fourier basis of $L^2(\Omega)$, we set for all $\ff \in L^1(\Omega \times I_L)$
    \begin{equation}\label{eq:Derg}
        \cD_{\rm erg}(\ff,\ff):=\frac{4\pi}{|\Omega|} \sum_{(G_1,G_2,k) \in \L_1^* \times \L_2^* \times \N^*}  \frac{|\widehat \ff_{G_1,G_2,k}|^2}{|G_1+G_2|^2+\lambda_k},
    \end{equation}
    where 
    $$
    \widehat \ff_{G_1,G_2,k}:= \int_{\Omega \times I_L} \ff(\omega,z) \, \overline{e_{G_1,G_2}}(\omega) \, g_k(z) \, \rd \omega \, \rd z.
    $$

\begin{lemma} \label{lem:Derg_sci} The map $L^1(\Omega \times I_L) \ni \ff \mapsto \cD_{\rm erg}(\ff,\ff) \in \R_+ \cup \{+\infty\}$ is convex. It is finite and continuous on $L^{r,p}$ for all $\tfrac32 < p \le \infty$ and $2 \le r \le \infty$, and weakly lower-semicontinuous on $L^p(\Omega \times I_L)$ for all $1 < p < \infty$. 
In addition, for all $0 \le \ff, \fg \in L^1(\Omega \times I_L)$ satisfying $\cD_{\rm erg}(\ff, \ff) < \infty$ and $\cD_{\rm erg}(\fg, \fg) < \infty$, we have
\[
    \cD_{\rm erg}(\ff, \fg) = \int_{I_L} \fint_{\Omega} \cV_\ff (\omega,z) \fg (\omega,z) \, \rd \omega \rd z 
    = \int_{I_L} \fint_{\Omega} \ff (\omega,z) \cV_\fg (\omega,z) \, \rd \omega \rd z.
\]

\end{lemma}
\begin{proof} It follows from~\eqref{eq:Derg} that $\cD_{\rm erg}$ is a convex positive quadratic form on $L^1(\Omega \times I_L)$, and weakly lower-semicontinuous on $L^p(\Omega \times I_L)$ for all $1 < p < \infty$ by Fatou's lemma. To prove continuity on $L^{r,p}$ for suitable $r,p$, we recall that the dual space of $L^{r,p}$ is $L^{r', p'}$ with $\frac{1}{r} + \frac{1}{r'} = 1$ and $\frac{1}{p} + \frac{1}{p'} = 1$. So, using that $p< \infty$ and $r' \le 2 \le r$, we get
    \begin{equation*}
    \cV_\ff \in L^{r, \infty} \subset L^{r, p'} \subset L^{r', p'} = (L^{r,p})^*. 
    \end{equation*}
    In particular, we have the inequality
    \begin{equation*}
    \cD_{\rm erg}(\ff, \ff) \le C \| \ff \|_{L^{r,p}} \| \cV_\ff \|_{L^{r, \infty}} \le C' \| \ff \|_{L^{r,p}}^2,
    \end{equation*}
    so the quadratic form $\ff \mapsto \cD_{\rm erg}(\ff, \ff)$ is strongly continuous for the $L^{r,p}$ topology, for all $\tfrac32 < p \le \infty$ and $2 \le r \le \infty$. To prove the equality when $0 \le \ff, \fg \in L^1(\Omega \times I_L)$, we note that we have equality for $\ff_N := \min \{ N, \ff\}$ by Lemma~\ref{lem:definition_Hartree_stationary}, and conclude by monotone convergence. Since $\cD_{\rm erg}$ is a positive quadratic form, we have the Cauchy-Schwarz inequality
    \[
        \cD_{\rm erg}(\ff, \fg)^2 \le \cD_{\rm erg}(\ff, \ff) \, \cD_{\rm erg}(\fg, \fg). \qedhere
    \]
\end{proof}

    \begin{lemma}
        \label{lem:definition_Hartree_stationary}
        Let $f_\omega(x, z)$ be a stationary function with generating function $\ff \in L^{r,p}$ with $\tfrac32 < p \le \infty$ and $2 \le r \le \infty$, and let $V_{f_\omega}$ be its associated potential, with generating function $\cV_\ff \in L^{r, \infty}$. Then, for almost all $\omega \in \Omega$, the Hartree energy per unit area
        \begin{equation}\label{eq:Hartree_energy_stat}
        \lim_{R \to \infty} \int_{I_L} \fint_{B_R} V_{f_\omega}(x, z) f_\omega(x, z) \rd x \rd z
        \end{equation}
        is well defined. 
        \begin{enumerate}
            \item If the two lattices $\L_1$ and $\L_2$ are {\bf incommensurate}, then this energy is almost surely independent of $\omega$, and we have for almost all $\omega \in \Omega$,
            \begin{equation}
            \lim_{R \to \infty} \int_{I_L} \fint_{B_R} V_\omega(x, z) f_\omega(x, z) \rd x \rd z = \int_{I_L} \fint_{\Omega} \cV_\ff(\omega', z) \ff(\omega', z) \rd \omega' \rd z=\cD_{\rm erg}(\ff, \ff). \label{eq:Derg_2}
            \end{equation}
            \item If the two lattices $\L_1$ and $\L_2$ are {\bf commensurate}, $f_\omega$ is $\L $-periodic, with $\L=\L_1\cap \L_2$, for each fixed configuration $\omega \in \Omega$, and we are back to the case of periodic material. However, there is no relation between the different configurations, and the Hartree energy depends in general on the configuration $\omega$.
        \end{enumerate}
    \end{lemma}

\begin{proof}
The first equality in~\eqref{eq:Derg_2} is a direct consequence of the Birkhoff--Von Neumann ergodic theorem, and we skip it for brevity. A direct computation shows that the second equality holds for $\ff \in L^2(\Omega \times I_L)$. We conclude by density. 
\end{proof}

    The original equation $- {\rm div}(\epsilon \nabla V_{f_\omega}) = 4 \pi f_\omega$ can be recast in terms of the generating functions $\ff$ and $V_\ff$. Indeed, Since $f_\omega(x, z) = \ff(\omega_1 - x, \omega_2 - x, z)$, we have
    \begin{equation} \label{eq:def:Lambda}
        (\nabla f_\omega)(0, z) = (\Lambda \ff) (\omega_1, \omega_2, z), \qquad \text{with} \quad
        \Lambda = \begin{pmatrix} 
        -\partial_{\omega_{1,1}} - \partial_{\omega_{2, 1}} \\ 
        -\partial_{\omega_{1,2}} - \partial_{\omega_{2, 2}} \\ 
        \partial_z \end{pmatrix}.
    \end{equation}
    The usual gradient operator therefore becomes the $\Lambda$ operator for the corresponding generating functions.  The two first components of $\Lambda$ are also $\nabla_{\omega_1}+\nabla_{\omega_2}$, and describe a derivation operator in a specific direction, induced by the group action $\alpha$. In particular, the operator $\Lambda^* \Lambda$ is not coercive. Nevertheless, we deduce that 
    \begin{equation*}
        \Lambda^* (\epsilon \Lambda \cV_\ff) = 4 \pi \ff \quad \text{on} \quad \Omega \times I_L.
    \end{equation*}
    Multiplying by $\cV_\ff$ and integrating shows that (compare with~\eqref{eq:def:Hartree_Periodic})
    \begin{equation*}
        \cD_{\rm erg}(\ff, \ff) = \int_{I_L} \fint_{\Omega} \cV_\ff(\omega', z) \ff(\omega', z) \rd \omega' \rd z = \frac{1}{4 \pi} \int_{I_L} \fint_{\Omega} \epsilon(z) | \Lambda \cV_\ff |^2(\omega, z) \rd \omega \, \rd z \qquad \ge 0.
    \end{equation*}

    \section{Reduced Hartree--Fock models for encapsulated quasi-periodic 2D materials}
    \label{sec:quasi-periodic}
    
    In this section, we introduce a reduced Hartree--Fock (rHF) model for encapsulated quasi-periodic 2D materials, prove the existence of ground-state density matrix and the uniqueness of the ground-state density, and derive the corresponding rHF equations.
    
        \medskip
    
    The rHF model is convex and strictly convex in the density, which greatly simplifies its analysis. It has been extensively studied in the mathematical literature~\cite{lieb1977hartreefock,solovej1991proof,catto1993binding,catto1998mathematical,catto2001thermodynamic,cances2008new,cances2013meanfield,gontier2016convergence}. In the periodic setting (Section~\ref{sec:periodic}), convexity is not so essential, and existence results can be established for non-convex models actually used in practice in condensed-matter physics, such as Kohn--Sham LDA.
    
    \medskip
    
  In what follows, we consider two incommensurate lattices $\L_1$ and $\L_2$, the ergodic group action $\R^2 \ni x \mapsto \alpha : \Omega \to \Omega$ on $\Omega = (\R^2 / \L_1) \times (\R^2 / \L_2)$ defined by $\alpha_x(\omega_1, \omega_2) = (\omega_1 - x, \omega_2-x)$, and a nuclear charge distribution $m : \Omega \times S_L \to \R$, stationary in the sense that
    \begin{equation*}
        \forall \omega \in \Omega, \quad m_\omega(x, z) = \fm(\alpha_x \omega, z),
    \end{equation*}
    for some generating function $\fm : \Omega \times I_L \to \R$. Our goal is to define an rHF energy functional for a stationary one-body density operator $\gamma$ describing electrons interacting with the nuclear charge distribution $m$, in the encapsulated setting.

    \subsection{Stationary operators and trace per unit area}
    \label{sec:stationary_operators}
    
    We first need to extend the notion of stationary functions to operators. We introduce the von Neumann algebra of (bounded) {\bf  fibered stationary operators} 
        \begin{equation} \label{eq:def:cA}
        \cA := \left\{ A = \int_\Omega^\oplus A_\omega\rd\omega \in L^\infty(\Omega,\cB(L^2(S_L))) \text{ such that } \forall x\in\R^2, U_x A_{\omega} U_x^* = A_{\alpha_{-x}\omega} \right\},
    \end{equation}
    where $U_x$ denotes the usual translation operator on $L^2(S_L)$ defined by $U_{x}(f)(x', z) = f(x' - x, z)$. This algebra can in fact~\cite{tak2001theoryI, tak2010theoryII} be identified with the tensor product
    \begin{equation*}
    \cA = \left(L^\infty(\Omega)\rtimes_\alpha \R^2\right) \, \overline{\otimes}\,\mathcal B(L^2(I_L)).
    \end{equation*}
    Such operators act on the Hilbert space $L^2(\Omega, L^2( S_L))$ as follows: if $f(\omega, x, z) =: f_\omega(x, z)$ is in $L^2(\Omega, L^2( S_L))$, then 
    \begin{equation*}
    A f = g \quad \text{means} \quad A_\omega f_\omega = g_\omega \quad \text{in $L^2(S_L)$, for a.a. $\omega \in \Omega$}.
    \end{equation*}
   A fibered self-adjoint operator $A=\int^\oplus_\Omega A_\omega \, \rd \omega$ is called {\em affiliated} to $\cA$ if its resolvent is in $\cA$. A typical example of self-adjoint operator affiliated to $\cA$ is the ergodic Schrödinger operator $H = \int_\Omega^\oplus H_\omega \, \rd \omega$ with $H_\omega := - \frac 12 \Delta_0 + V_\omega$, where $\Delta_0$ is the Dirichlet Laplacian on $S_L$ and $V_\omega(x) = V(\omega, x)$ a stationary potential in a suitable function space. 
    
    \medskip
    
    The von Neumann algebra $\cA$ is a type $\mathrm{II}_\infty$ factor~\cite[Chap. XIII]{tak2010theoryIII} and therefore has a unique semifinite faithful normal trace (up to a scaling factor), which we denote by $\tau$. In what follows, we set, for $1 \le p < \infty$,
    \begin{equation} \label{eq:def:ufS1}
        L^p(\cA, \tau) := \{ A \in \widetilde{\cA} \ : \ \tau(|A|^p)<\infty \}, \quad \text{and} \quad 
        L^\infty(\cA, \tau) := \cA.
    \end{equation}
    The non-commutative $L^p(A, \tau)$ spaces share properties which are similar to the usual Lebesgue spaces, see~\cite{segal1953noncommutative,nelson1974notes}. The following result highlights the fact that one can fix the scaling factor in $\tau$ so that it corresponds to the notion of trace per unit area. For a function $\eta \in L^\infty(\R^2)$, we denote by $M_\eta$ the multiplication operator by $\eta(x)$ (independent of $z$).
    
    \begin{proposition} \label{prop:ergodic-trace}
        Let $A\in L^1(\cA, \tau)$. Then, for a.e. $\omega\in\Omega$, $A_\omega$ is locally trace-class on $L^2(S_L)$. If $\rho_{A_\omega}$ denotes the density of $A_\omega$, then $\rho(\omega, x, z) := \rho_{A_\omega}(x, z)$ is a stationary function. The corresponding generating function $\varrho_A$ belongs to $L^1(\Omega \times I_L)$, and is called the {\bf generalized density}  of $A$.
        
        \medskip
        
        For any $\eta\in L^\infty(\R^2) \cap L^2(\R^2)$, with $\| \eta \|_2 = 1$, we have $\int_\Omega \Tr_{L^2(S_L)}(|M_\eta A_\omega M_\eta|)\rd\omega<\infty$, and the formulas
        \begin{equation}
            \label{eq:trace_per_unitV_formula}
            \tau(A) = \fint_\Omega \Tr_{L^2(S_L)}(M_\eta A_\omega M_\eta)\rd\omega, \quad
            \text{and} \quad
            \tau(A)  = \int_{I_L} \fint_\Omega \varrho_A(\omega, z) \rd \omega \rd z.
        \end{equation}
    \end{proposition}
    
    \begin{proof}
        Denote by $\rho_{\omega}$ the density of $A_\omega$, so that
        \begin{equation} \label{eq:def:rho}
            \forall \eta \in C^\infty_0(S_L), \quad \Tr_{L^2(S_L)}( M_\eta A_\omega M_\eta) = \int_{S_L} \rho_\omega | \eta |^2.
        \end{equation}
        Since $U_x A_\omega U_x^* = A_{\alpha_{-x} \omega}$, we have $\rho_\omega(x' - x, z) = \rho_{\alpha_{-x}(\omega)}(x', z)$, so $\rho(\omega, x, z) := \rho_\omega(x,z)$ is stationary. In particular, since $\alpha_x$ is measure preserving, the quantity
        \begin{equation*}
            \int_{I_L} \fint_\Omega \rho(\omega, x, z) \rd \omega \rd z 
            = \int_{I_L} \fint_\Omega \rho(\alpha_x \omega, 0, z) \rd \omega \rd z 
            = \int_{I_L} \fint_\Omega \rho(\omega, 0, z) \rd \omega \rd z 
            = \int_{I_L} \fint_\Omega \varrho(\omega, z) \rd \omega \rd z
        \end{equation*}
        is independent of $x \in \R^2$. We define $\tau(A)$ the trace per unit surface of $A$ as this common quantity. We can check that $\tau(\cdot)$ is linear and that $\tau(A) \ge 0$ if $A \ge 0$, so this is the von Neumann trace (up to a scaling factor). We now integrate~\eqref{eq:def:rho} in $\omega \in \Omega$, and get that
        \begin{equation*}
            \fint_{\Omega} \Tr_{L^2(S_L)}( M_\eta A_\omega M_\eta)\rd\omega = \int_{S_L} \left( \fint_{\Omega}  \rho(\omega, x, z) \rd \omega \rd z \right) | \eta |^2 (x) \rd x = \tau(A) \int_{\R^2} | \eta |^2(x) \rd x = \tau(A).
        \end{equation*}
    \end{proof}
    
    Taking $\eta = |B_R|^{-1/2}\mathds{1}_{B_R}$ and the limit $R\to +\infty$, we get
    \begin{equation*}
        \tau(A) = \lim_{R\to +\infty} \fint_{\Omega} \frac{1}{|B_R|}\Tr_{L^2(S_L)}(M_{\mathds{1}_{B_R}}A_\omega M_{\mathds{1}_{B_R}}) \rd \omega
        = \lim_{R \to \infty} \fint_\Omega \int_{I_L} \fint_{B_R} \rho_{A_\omega}(x, z) \rd x \rd z \rd \omega.
    \end{equation*}
    In the ergodic case, Birkhoff theorem implies that we also have
    \begin{equation*}
        \tau(A) = \lim_{R\to +\infty} \frac{1}{|B_R|}\Tr (M_{\mathds{1}_{B_R}}A_\omega M_{\mathds{1}_{B_R}})
        = \lim_{R \to \infty} \int_{I_L} \fint_{B_R} \rho_{A_\omega}(x, z) \rd x \rd z \quad \mbox{for a.a. } \omega \in \Omega.
    \end{equation*}
    
    \begin{example}[Fourier multipliers]
        \label{ex:FourierMultiplier}
         As an example, consider a convolution operator $J$ acting on $L^2(\R^2)$ with kernel $J(x, x') = j(x - x')$. Then $J$ is a Fourier multiplier by $\widehat{j}(k)$, and we can write $J = \widehat{j}( - \ri \nabla_x)$. We set
        \begin{equation*}
            A := \int_\Omega^\oplus A_\omega \rd \omega, \quad \text{with} \quad A_\omega := J \otimes  | h \rangle \langle h | ,
        \end{equation*}
        where $h \in L^2(I_L)$ is normalized $\| h \|_{L^2(I_L)} = 1$. The function $h$ in the $z$-variable is only here to handle the $z$-direction. Since $A_\omega$ is independent of $\omega$ and commutes with all translations, $A$ is a fibered stationary operator affiliated to $\cA$. A simple computation shows that $\varrho_A(\omega, z) = j(0) | h(z) |^2$. In particular, $\varrho_A$ is independent of $\omega \in \Omega$. Similarly, we get
        \begin{equation*}
        \tau(| A |^p) = \frac{1}{(2 \pi)^2} \int_{\R^2} | \widehat{j}(k) |^p \rd k.
        \end{equation*}
        We deduce that $A \in L^p(\cA, \tau)$ iff $\widehat{j} \in L^p(\R^2)$. In particular, $A \in L^p(\cA, \tau)$ is not necessarily a bounded operator in $\cA = L^\infty(\cA, \tau)$. Note that the usual Schatten classes $\fS_p := L^p(\cB(L^2), \Tr )$ consist of compact -hence bounded- operators.
    \end{example}

    \subsection{The encapsulated kinetic energy}
    \label{ssec:kinetic_energy}
     We also define the subspace $W_0^{1,1}(\cA, \tau)$ of stationary operators of finite kinetic energy (per unit area) by
    \begin{equation} \label{eq:def:W11_0}
        W^{1,1}_0(\cA, \tau) := \left\{ A \in L^1(\cA, \tau) \ | \ \sqrt{- \Delta_0} | A |\sqrt{- \Delta_0} \in L^1(\cA, \tau) \right\},
    \end{equation}
    where $-\Delta_0$ is the Dirichlet Laplacian on $S_L$, that is with domain $H^2(S_L) \cap H^1_0(S_L)$ (we put a $0$ subscript to emphasize Dirichlet boundary conditions). With a slight abuse of notation, we extended the operator $-\Delta_0$ as an affiliated fibered stationary operator by setting $\Delta _0 =\int_\Omega^\oplus \Delta_0 \rd\omega$ on $L^2(\Omega \times S_L)$, that is with an action independent of the configuration $\omega \in \Omega$. 
    
    \medskip
    
    The condition in~\eqref{eq:def:W11_0} is equivalent to $\sqrt{- \Delta_0} | A |^{1/2} \in L^2(\cA, \tau)$. It implies that all eigenspaces of $A$ are in the form domain of $- \Delta_0$, hence satisfy Dirichlet boundary conditions. It implies for instance that its generalized density $\varrho_A$ satisfies the boundary conditions
    \begin{equation*}
           \varrho_A\left( \cdot, \pm \tfrac{L}{2} \right) = 0 \quad \mbox{in the sense of the trace}.
    \end{equation*}
    For $A \in W^{1,1}_0(\cA, \tau)$, we define the \textbf{encapsulated kinetic energy per unit area} of $A$ by 
    \begin{equation}
        \label{eq:def:kin_energy_bulk}
        \frac12\tau(-\Delta_0 A) := \frac12\tau \left( \sqrt{- \Delta_0} A \sqrt{- \Delta_0} \right).
    \end{equation}
    Using the spectral decomposition of the 1D Dirichlet Laplacian given by 
    \begin{equation*}
        -\partial_{zz,0}^2 = \sum_{k=1}^\infty \mu_k | h_k \rangle \langle h_k | , \quad \mu_k := \frac{\pi^2 k^2}{L^2}, \quad h_k(z) := \sqrt{\frac{2}{L}} \sin \left(\pi k \left[\frac{z}{L} + \frac12 \right] \right),
    \end{equation*}
    we have
    \begin{equation} \label{eq:decomposition_kinetic}
        - \Delta_0 = ( - \Delta_x) \otimes \bbI_{L^2(I_L)} + \bbI_{L^2(\R^2)} \otimes (- \partial_{zz,0}^2) 
        = \sum_{k=1}^\infty ( - \Delta_x + \mu_k) \otimes | h_k \rangle \langle h_k |.
    \end{equation}
    Again, we see the effect of the boundary conditions, which implies the spectral gap $(- \Delta_0) \ge \mu_1 > 0$, as in Lemma~\ref{lem:spectral_gap}.

    \subsection{Stationary one-body density matrices, and kinetic inequalities}
    In what follows, we are interested in \textbf{stationary one-body density matrices}. We therefore set
    \begin{equation*}
        \cK := \left\{ \gamma \in W^{1,1}_0(\cA, \tau), \quad  \ 0\leq\gamma = \gamma^*\leq 1\right\}.
    \end{equation*}
    If we write $\gamma = \int_\Omega^\oplus \gamma_\omega \rd \omega$ with $\gamma_\omega$ acting on $L^2(S_L)$, then we have
    \begin{equation*}
    0 \le \gamma_\omega = \gamma_\omega^* \le 1, \quad \text{and} \quad
    \forall x \in \R^2, \quad U_x \gamma_\omega U_x^* = \gamma_{\alpha_{-x} \omega},
    \end{equation*}
    where $U_x$ is the longitudinal-translation operator introduced in Section~\ref{sec:stationary_operators}. Proposition~\ref{prop:ergodic-trace} above shows that any such one-body density matrix admits a generalized density $\varrho_\gamma \in L^1(\Omega \times I_L)$. We record several inequalities for stationary operators, including Lieb--Thirring inequality~\cite{lieb1975bound,lieb1976inequalities}, and Hoffmann--Ostenhof inequality~\cite{hoffmann-ostenhof1977schrodinger}. Recall that the operator $\Lambda$ has been defined in~\eqref{eq:def:Lambda}.
    \begin{lemma} \label{lem:inequalities}
        For all $\gamma \in \cK$, we have
        \begin{align}
            & \tau( - \Delta_0 \gamma) \ge \frac{\pi^2}{L^2} \tau(\gamma) = \frac{\pi^2}{L^2} \int_{I_L} \fint_{\Omega} \varrho_\gamma & \text{(Poincaré's inequality)},  \label{eq:Poincare_op} \\
           & \tau( - \Delta_0 \gamma) \ge C_{\rm LT} \int_{I_L} \fint_{\Omega} \varrho_\gamma^{5/3} & \text{(Lieb--Thirring inequality)}, \label{eq:Lieb-Thirring}  \\
           & \tau( - \Delta_0 \gamma) \ge 
        \int_{I_L} \fint_{\Omega} |\Lambda \sqrt{\varrho_\gamma} |^2  & \text{(Hoffmann--Ostenhof inequality)}.\label{eq:Hoffmann-Ostenhof}
        \end{align}
    \end{lemma}
    The kinetic energy of $\gamma$ therefore controls the $L^p$-norm of $\varrho_\gamma$ for all $1 \le p \le 5/3$, by interpolation. Note that the Poincaré inequality is specific to our encapsulated setting, while the two other inequalities also hold for standalone 2D materials.
        
    \begin{proof}
        Poincaré's inequality comes from $\gamma \ge 0$ together with the operator inequality $-\Delta_0 \ge \mu_1$.

    In order to prove the Lieb--Thirring inequality in the encapsulated stationary case, we use two tools. To handle the boundary conditions in the $z$-direction, we introduce the extension-by-$0$ operator $i : L^2(S_L) \to L^2(\R^3)$, and set
    \begin{equation*}
        \widetilde{\gamma}_\omega := i  \gamma_\omega  i^*, \qquad \text{acting on $L^2(\R^3)$.}
    \end{equation*} 
    Since $\gamma_\omega$ satisfies Dirichlet boundary conditions, we have 
    \begin{equation*}  
        \sqrt{ - \Delta_0 } \gamma_\omega \sqrt{ - \Delta_0 } = \sqrt{ - \Delta }  \widetilde{\gamma}_\omega \sqrt{ - \Delta},
    \end{equation*}
    where $- \Delta$ denotes the usual three--dimensional Laplacian.
    
    \medskip
    
    The second tool we use is  a thermodynamic procedure to control the fact that the operators $\gamma_\omega$ and $\widetilde{\gamma}_\omega$ do not decay at infinity in the longitudinal directions. For $R > 0$, we set $\chi_R(x, z) := \chi \left( \frac{x}{R} \right)$ (independent of $z$), where $\chi : \R^2 \to [0, 1]$ is a smooth cut-off radial nonincreasing function, with $\chi(x) = 1$ for $| x | < 1$ and $\chi(x) = 0$ for $| x | > 2$. We set
    \begin{equation*}
        \gamma_{\omega, R} := \chi_R \widetilde{\gamma}_\omega \chi_R = \chi_R  i \, \gamma_\omega   \, i^* \chi_R.
    \end{equation*} 
    The operator $\gamma_{\omega, R}$ satisfies $0 \le\gamma_{\omega, R} \le 1$ and is trace class. We denote by $\rho_{\omega, R}$ its density. The usual Lieb--Thirring inequality~\cite{lieb1975bound, lieb1976inequalities} shows that $\rho_{\omega, R} \in L^{5/3}(\R^3)$ with compact support, and that
    \begin{equation*}
    C_{\rm LT} \int_{S_L} \rho_\omega^{5/3} \chi_R^{10/3} = C_{\rm LT} \int_{\R^3} \rho_{\omega, R}^{5/3} \le   \Tr_{L^2(\R^3)} \left(  - \Delta \gamma_{\omega, R} \right) := \sum_{j=1}^3 \Tr \left( P_j \gamma_{\omega, R} P_j \right),
    \end{equation*}
    where $P_j := - \ri \partial_{x_j}$ the momentum operator along the $x_j$-direction. We have, for $j \in \{ 1, 2 \}$, 
    \begin{equation*}
    P_j \gamma_{\omega, R} P_j = \chi_R (P_j \widetilde{\gamma}_\omega P_j) \chi_R + [P_j, \chi_R] \widetilde{\gamma}_\omega \chi_R + \chi_R \widetilde{\gamma}_\omega [P_j, \chi_R],
    \end{equation*}
    with $[P_j, \chi_R] = (- \ri \partial_j \chi_R)(x) = \frac{1}{R}  (- \ri \partial_j \chi)(x/R)$. Integrating in $\omega \in \Omega$ and taking the trace shows that 
    \begin{equation*}
     \fint_\Omega \Tr \left( P_j \gamma_{\omega, R} P_j \right) = \fint_\Omega \Tr \left(  \chi_R (P_j \widetilde{\gamma}_\omega P_j) \chi_R \right)  + O \left( R \right).
    \end{equation*}
    Using Proposition~\ref{prop:ergodic-trace} with $\eta := \frac{\chi_R}{\| \chi_R \|_2}$ gives
    \begin{equation*}
    C_{\rm LT}  \left( \int_{I_L} \fint_\Omega  \varrho^{\frac{5}{3}} \right) \frac{ \| \chi_R \|_{\frac{10}{3}}^{\frac{10}{3}}}{\| \chi_R \|_2^2} \le \tau( - \Delta_0 \gamma ) + O (R^{-1}).
    \end{equation*}
    Note that the ratio in the left-hand side is independent of $R$. Taking $R \to \infty$ shows that the remainder is actually null. For the previous computation to be valid, we need $\chi$ to be differentiable. However, since the result is valid for all $\chi$, one can consider a sequence $\chi_n$ converging to $\1(| x | < 1)$, and we get, as wanted,
    \begin{equation*}
        C_{\rm LT}  \left( \int_{I_L} \fint_\Omega   \varrho^{\frac{5}{3}} \right) \le \tau( - \Delta_0 \gamma ),
    \end{equation*}
    with the same 3D Lieb--Thirring constant.
     
     \medskip

     We proceed similarly for the proof of Hoffmann--Ostenhof inequality. We get
     \begin{equation} \label{eq:HO}
         \Tr (- \Delta \gamma_{\omega, R}) \ge \int_{\R^3}  | \nabla \sqrt{\rho_{\omega, R}} |^2 .
     \end{equation}
     We conclude as for the proof of the Lieb--Thirring inequality. The passage from $| \nabla \sqrt{\rho} |$ to $| \Lambda \sqrt{\varrho} |$ comes from~\eqref{eq:def:Lambda}.
    \end{proof}
    
      \begin{remark}[Lieb--Thirring inequality for Neumann boundary conditions]
            The bound \eqref{eq:Lieb-Thirring} does not hold if we replace the Dirichlet boundary conditions by the Neumann ones. Indeed, consider a smooth compactly supported 2D function $j \in C^\infty_0(\R^2)$ with $0 \le \widehat{j}(k) \le 1$ pointwise. For $\lambda > 0$ a scaling factor, we denote by $J_\lambda$ the convolution operator acting on $L^2(\R^2)$, and with kernel $J_\lambda(x, x') = \lambda^2 j(\lambda (x - x'))$. The scaling is chosen so that $J_\lambda$ is the Fourier multiplier by $\widehat{j}(k/\lambda)$. We finally consider the operator (see Example~\ref{ex:FourierMultiplier})
            \begin{equation*}
            \gamma_\lambda := J_\lambda \otimes \left( L^{-1} | \1 \rangle \langle \1 | \right).
            \end{equation*}
            Since $0 \le \widehat{j}(k/ \lambda) \le 1$, we have $0 \le \gamma_\lambda = \gamma_\lambda^* \le 1$. A computation shows that the density of $\gamma_\lambda$ is constant, independent of $(x, z) \in S_L$, given by
            \begin{equation*}
            \rho_\lambda(x) = \gamma_\lambda(x, x) = \lambda^2 \rho_{\lambda = 1}(x) = \frac{1}{L}\lambda^2 j(0).
            \end{equation*}
            On the other hand, the kinetic energy per unit area scales as the two-dimensional one, namely
            \begin{equation*}
            \tau ( - \Delta_{\rm N} \gamma_\lambda) = \lambda^4 \tau( - \Delta_{\rm N} \gamma_{\lambda = 1}) = \lambda^4 \frac{1}{(2 \pi)^2} \int_{\R^2} | k |^2 \widehat{j}(k) \rd k.
            \end{equation*}
            The scaling in $\lambda$ of $\tau ( - \Delta_{\rm N} \gamma_\lambda) $ therefore matches the one of $\rho_\lambda^2$ rather than $\rho_\lambda^{5/3}$. We recover the 2D Lieb--Thirring inequality in this example.
        \end{remark}

    \subsection{Reduced Hartree--Fock model}
    
    We define the rHF energy (per unit area) functional of a trial stationary density-matrix $\gamma \in \cK$ as
    \begin{equation}\label{eq:rhf_energy_functional}
    \boxed{ \cE^{\rm rHF}(\gamma) := \frac12 \tau(- \Delta_0 \gamma) + \frac 12 \cD_{\rm erg}(\varrho_\gamma - \fm, \varrho_\gamma - \fm) + v_0 \fint_{\Omega} \int_{I_L} g_0 \varrho_\gamma.}
    \end{equation}
    Here $\fm$ is the generating function of the charge density $m$, that is $\fm(\alpha_x\omega,z)=m(\omega,x,z)$. The first term is the kinetic energy introduced in Section~\ref{ssec:kinetic_energy}, and the second term is the Hartree energy, introduced in Section~\ref{ssec:Coulomb_stationary}. The last term is the potential energy induced by the gating potential $v_0$ between the two electrodes, see~\eqref{eq:Vv0-V0}.
    
        \medskip
    
    Our main result is contained in the following theorem. For brevity, we focus on the microcanonical formulation of the rHF ground-state problem, but similar results can be obtained for the grand-canonical formulation
    $$
     \min \left\{ \cE^{\rm rHF}(\gamma)-\mu \tau(\gamma), \; \gamma \in \cK \right\},
     $$
     where $\mu \in \R$ is a given arbitrary chemical potential.
    
    \begin{theorem} \label{th:main_KS_moire}
        Let $\fm \in L^{\infty,p}$, with $\frac{3}{2}<p\le \infty$ and $\lambda > 0$. The minimization problem
        \begin{equation}\label{eq:min_rHF_quasi-periodic}
        \min \left\{ \cE^{\rm rHF}(\gamma), \; \gamma \in \cK, \ \tau(\gamma) = \lambda \right\}
       \end{equation}
        admits a minimizer and all the minimizers share the same density $\rho_*$. Let $\gamma_*$ be a minimizer of \eqref{eq:min_rHF_quasi-periodic}, and $\varrho_*$ the generating function of $\rho_*$. Then $\varrho_* \in L^\infty(\Omega \times I_L)$. In addition, for almost all $\omega \in \Omega$, we have the Euler--Lagrange equations
\begin{equation}\label{eq:characterization_minimizers}
        \gamma_{*,\omega} = \1 (H_{*,\omega} < \eps_F) + \delta_\omega,
        \end{equation}
        where $\varepsilon_F\in\R$ is some Fermi level, i.e. the Lagrange multiplier of the constraint $\tau(\gamma) = \lambda$, and $0 \le \delta_\omega \le \1(H_{*,\omega} = \eps_F)$. Here, $H_{*,\omega}$ is the mean-field operator, defined by
        \begin{equation*}
        H_{*,\omega} := -  \frac 12 \Delta_0 + V_{\rho_{*,\omega} - m_\omega} + v_0 g_0.
        \end{equation*}
         where $V_{\rho_{*,\omega} - m_\omega} \in L^\infty(S_L)$ is the encapsulated mean-field potential generated by the total charge density $\rho_{*,\omega} - m_\omega$, as defined in Section~\ref{ssec:Coulomb_stationary}. Finally, it holds that the generating function $\cV_{\varrho_*-\fm}$ is in $L^\infty(\Omega \times I_L)$ as well.
    \end{theorem}
    
    Note that the assumption $\fm\in L^{\infty,p}$ with $p>\frac{3}{2}$ is quite natural, in the context of smeared nuclei. Indeed, this assumption is satisfied if the atomic charge density $m_i$ corresponding to the lattice $\L_i$ belongs to $L^p_\unif(S_L)$, $p>\frac{3}{2}$. Recall that $\omega\in\Omega$ corresponds physically to a translation in space of the lattice. From~\eqref{eq:atomic_charge_density}, it follows that $m_\omega\in L^p_\unif(S_L)$ as well, and $\|m_\omega\|_{L^p_\unif(S_L)}$ is bounded uniformly in $\omega$, hence the generating function verifies $\fm\in L^{\infty,p}$.

    \medskip

    \begin{remark} Theorem~\ref{th:main_KS_moire} shows that the ground-state problem~\eqref{eq:min_rHF_quasi-periodic} has a minimizer for {\em any} positive value $\lambda$ of the electronic charge per unit area. However, the model only makes sense when the electronic density of the moir\'e material penetrates the dielectric medium only slightly, and decays very fast away from the moiré material. This is expected to be the case around charge neutrality, which is the physically most relevant case (and confirmed with preliminary numerical simulations). This also motivates the use of Dirichlet boundary conditions in our definition of the kinetic energy operator $-\frac 12 \Delta_0$. When $\lambda$ increases, the ground-state density starts to penetrate the dielectric continuum modeling the insulating substrate, a phenomenon reminiscent of the escaped charge problem encountered in quantum chemistry when using implicit solvent models for describing solvated ions~\cite{CM2001}. For $\lambda$ too large, the implicit assumption underlying the model, that the moir\'e material electrons can somehow be decoupled from the insulating substrate electrons, breaks down.
    \end{remark}
    
    \begin{proof}[Proof of Theorem~\ref{th:main_KS_moire}]
    
    We split the proof into several steps for clarity.
    
    \medskip
    
    \noindent \underline{Step 1: The minimization set is non empty}. 
    Let us first prove that one can find $\gamma \in \cK$ with $\tau(\gamma)= \lambda$ and $\cE^{\rm rHF}(\gamma) < \infty$. Consider a smooth function $h \in H^1_0(I_L)$ with $\| h \|_{L^2(I_L)} = 1$, a smooth function $j : \R^2 \to \R$ satisfying $\widehat{j}(k) \in C^\infty_0(\R^2)$ and $0 \le \widehat{j}(k) \le 1$ pointwise, and the operator (see also Example~\ref{ex:FourierMultiplier})
    \begin{equation} \label{eq:example_non_empty}
    \gamma = \int_\Omega^\oplus \gamma_\omega \rd \omega, \quad \gamma_\omega := \widehat{j}( - \ri \nabla_x) \otimes | h \rangle \langle h |.
    \end{equation}
    Note that $\gamma_\omega$ is in fact independent of $\omega \in \Omega$, and acts as a Fourier multiplier in the longitudinal variables. Since $\gamma_\omega$ is translation-invariant in the longitudinal directions, $\gamma$ is a fibered stationary operator. The density of $\gamma$ is
    \begin{equation*}
    \rho_\gamma(x, z) = \left( \frac{1}{(2 \pi)^2} \int_{\R^2} \widehat{j}(k) \rd k \right) | h(z) |^2 = j(0) | h(z) |^2,
    \end{equation*}
    and in particular is independent of $x \in \R^2$. So 
    \begin{equation*}
    \varrho(\omega, z) = j(0) | h(z) |^2
    \quad \text{and} \quad
    \tau(\gamma) =  j(0) .
    \end{equation*}
    Since $h\in H^1_0(I_L)$, it is continuous by the Sobolev embedding, so that $\varrho$ is continuous as well. In particular, it is bounded, so that $\varrho - \fm \in L^{\infty,p}$. Since $p>\frac{3}{2}$ by assumption, the Hartree energy $\cD_{\rm erg}(\varrho - \fm, \varrho - \fm)$ is finite. Similarly, the kinetic energy of $\gamma$ is
    \begin{equation*}
    \frac{1}{2}\tau(- \Delta_0 \gamma) =  \frac{1}{2}\left( \frac{1}{(2 \pi)^2} \int_{\R^2} | k |^2 \widehat{j}(k) \rd k \right) + \frac{1}{2}j(0) \| \nabla h \|_{L^2(I_L)}^2,
    \end{equation*}
    and in particular is finite, since $\widehat{j} \in C^\infty_0(\R^2)$. This proves that the minimization set is non empty.
    
    \medskip
    
    \noindent \underline{Step 2: The energy $\cE^{\rm rHF}$~\eqref{eq:rhf_energy_functional} is bounded from below}.
    This follows directly from the fact that the first two terms are nonnegative, and the last one is bounded from below by $- |v_0| \, \| g_0 \|_{L^\infty} \| \varrho \|_{L^1}$, where $g_0(z)$ is defined in~\eqref{eq:def:g0}. Actually, since $g_0$ is strictly increasing with $g_0(\pm L/2) = \pm 1/2$, we have $\| g_0 \|_\infty = \frac12$. This gives $\cE^{\rm rHF} \ge - \frac12 | v_0 | \lambda$ on $\cK$.
    
    \medskip
    
    \noindent \underline{Step 3: Extraction of a subsequence}. Consider $(\gamma_n)$ a minimizing sequence for \eqref{eq:min_rHF_quasi-periodic}. As
    $$
    0 \le \frac{1}{2}\tau(- \Delta_0 \gamma_n) \le \cE^{\rm rHF} (\gamma_n) +   | v_0 | \, \| g_0 \|_{L^\infty} \lambda,
    $$
    the sequence $(\tau(- \Delta_0 \gamma_n))$ is bounded. Since $\gamma_n \in \cK$ satisfies $0 \le \gamma_n \le 1$ and $\tau(\gamma_n) = \lambda$, we deduce that the sequence $(\gamma_n)$ is bounded in $L^1(\cA, \tau) \cap L^\infty(\cA, \tau)$, hence in all $L^p(\cA, \tau)$ spaces.
    
    For $p > 1$, we may use the Banach--Alaoglu theorem. Up to a subsequence (undisplayed), we may assume that
    \begin{equation*}
    \begin{cases}
        \gamma_n \rightharpoonup \gamma_* \quad \text{weakly in $L^p(\cA, \tau)$ for all $1 < p < \infty$} \\
        \gamma_n \rightharpoonup_* \gamma_* \quad \text{weakly-$*$ in $L^\infty(\cA, \tau)$.}
    \end{cases}
    \end{equation*}
    In particular, we have $0 \le \gamma_* \le 1$ and $\tau(\gamma_*) \le \lambda$. Finally, Fatou's Lemma implies that
    \begin{equation*}
    \tau( - \Delta_0 \gamma_* ) = \tau( \sqrt{- \Delta_0} \gamma_* \sqrt{- \Delta_0}) \le \liminf_{n \to \infty} 
    \tau( \sqrt{- \Delta_0} \gamma_n \sqrt{- \Delta_0})  \le C
    \end{equation*}
    is finite. So $\gamma_* \in W^{1,1}_0(\cA, \tau)$.
    
    \medskip
    
    \noindent \underline{Step 4: Convergence of the densities.} We follow the approach in~\cite{cances2013meanfield}. This time, we use the bound
    \begin{equation*}
    \tau( \sqrt{- \Delta_0} \gamma_n \sqrt{- \Delta_0}) \le C.
    \end{equation*}
    Equivalently, $\sqrt{- \Delta_0} \gamma_n^{1/2}$ is bounded in $L^2(\cA, \tau)$. Let us focus on the sequence
    \begin{equation*}
    T_n := \sqrt{- \Delta_0} \gamma_n.
    \end{equation*}
    Since $T_n = \sqrt{- \Delta_0} \gamma_n = \sqrt{- \Delta_0} \gamma_n^{1/2} \gamma_n^{1/2}$ with  $\sqrt{- \Delta_0} \gamma_n^{1/2}$ bounded in $L^2(\cA, \tau)$ and $\gamma_n^{1/2}$ bounded in $L^2(\cA, \tau) \cap L^\infty(\cA, \tau)$, we deduce that $T_n$ is bounded in $L^1(\cA, \tau) \cap L^2(\cA, \tau)$. According to the Banach--Alaoglu theorem, we may extract a subsequence (undisplayed) and a limit $T_*$ so that
    \begin{equation*}
    T_n \rightharpoonup T_* \quad \text{in $L^p(\cA, \tau)$ for all $1 < p \le 2$}.
    \end{equation*}
    Finally, we set $\widetilde{\gamma}_* := \sqrt{- \Delta_0}^{-1} T_*$.
    
    \medskip
    
    Let us prove that $\gamma_n$ converges weakly to $\widetilde{\gamma}_*$ in $L^p(\cA, \tau)$ spaces for all $1 \le p < 2$ (this will eventually prove that $\widetilde{\gamma}_* = \gamma_*$). 
    \begin{lemma}
        We have  $(- \Delta_0)^{-1} \in L^q(\cA, \tau)$ for all $q > 3/2$, and
        \begin{equation*}
            \tau ((- \Delta_0)^{-q}) = \dfrac{ L^{2(q - 1)} }{4 (q - 1) \pi^{2q - 1 }} \zeta(2q - 2),
        \end{equation*}
        where $\zeta(\cdot)$ is the Riemann zeta function.
    \end{lemma}
    
    \begin{proof}
        From Eqn.~\eqref{eq:decomposition_kinetic}, we get
        \begin{equation*}
        - \Delta_0 = \sum_{k=1}^\infty ( - \Delta_x + \mu_k) \otimes | h_k \rangle \langle h_k |, \quad \text{hence} \quad
        \frac{1}{(- \Delta_0)^q} = \sum_{k=1}^\infty \dfrac{1}{( - \Delta_x + \mu_k)^q} \otimes | h_k \rangle \langle h_k |.
        \end{equation*}
        This is the sum of operators of the forms considered in Example~\ref{ex:FourierMultiplier}. We deduce that
        \begin{equation*}
        \tau \left(  \dfrac{1}{( - \Delta_x + \mu_k)^q} \otimes | h_k \rangle \langle h_k | \right) = \frac{1}{(2 \pi)^2} \int_{\R^2} \dfrac{\rd \xi}{(|\xi|^2 + \mu_k)^{q}} 
        = \frac{1}{(2 \pi)^2} \mu_k^{1 - q} \int_{\R^2} \underbrace{\dfrac{\rd \xi}{(|\xi|^2 + 1)^{q}}}_{ = \frac{\pi}{q - 1} }.
        \end{equation*}
        Recall that $\mu_k = \frac{\pi^2 k^2}{L^2}$. We obtain that for $q > \frac32$,
        \begin{equation*}
        \tau( ( - \Delta_0)^{-q} ) = \frac{1}{4 \pi (q - 1)} \sum_{k=1}^\infty \left( \frac{L^2}{\pi^2 k^2} \right)^{q - 1} = \dfrac{ L^{2(q - 1)} }{4 (q - 1)\pi^{2q - 1}} \sum_{k=1}^\infty \frac{1}{k^{2q - 2}} = \dfrac{ L^{2(q - 1)}}{4 (q - 1) \pi^{2q - 1}} \zeta(2q - 2),
        \end{equation*}
        which concludes the proof.
    \end{proof}
    
    Since $(- \Delta_0)^{-1/2}$ is in $L^q(\cA, \tau)$ for all $q > 3$, then for all $A \in L^r(\cA, \tau)$ with $r > 1$, we have $A (- \Delta_0)^{-1/2}$ in $L^{s}(\cA, \tau)$ for all $\frac{3r}{r+3} < s \le r$. In particular, for all $2 < r \le \infty$, we can tune $q$ so that $s > 2$ as well, in which case $A (- \Delta_0)^{-1/2}$ is a valid test operator for the weak-convergence of $T_n$.
    
    We deduce that, for all $A \in L^r(\cA, \tau)$ with $2 < r \le \infty$, we have
    \begin{equation*}
    \tau(A \gamma_n) = \tau(A (- \Delta_0)^{-1/2} T_n) \xrightarrow[n \to \infty]{} \tau(A (- \Delta_0)^{-1/2} T_*) = \tau( A \widetilde{\gamma}_*).
    \end{equation*}
    This proves that $\gamma_n$ converges weakly to $\widetilde{\gamma}_*$ in $L^p(\cA, \tau)$ for all $1 \le p < 2$. In particular, we have $\gamma_* = \widetilde{\gamma}_*$.
    
    \medskip

    We denote by $\varrho_n$ and $\varrho_*$ the generalized density of $\gamma_n$ and $\gamma_*$ respectively. Consider any positive bounded function $\Phi(\omega, z): \Omega \times I_L \to \R$, and set 
    \begin{equation*}
    M_\Phi := \int_\Omega (M_\phi)_\omega \rd \omega, \qquad [(M_\phi)_\omega u](x,z) := \Phi(\alpha_x \omega, z) u(x,z).
    \end{equation*}
    In other words, $M_\Phi$ represents the multiplication operator by a stationary function
    $\phi_\omega(x, z) :=  \Phi(\alpha_x \omega, z)$. Since $\phi \in L^\infty$, we have $M_\Phi \in L^\infty(\cA, \tau)$. Taking $A = M_\Phi$ in the previous computation shows that
    \begin{equation*}
    \int_{\Omega \times I_L} \Phi \varrho_n = \tau(M_\Phi \gamma_n) \xrightarrow[n \to \infty]{}  \tau( M_\Phi \gamma_*) =  \int_{\Omega \times I_L} \Phi \varrho_*.
    \end{equation*}
    We deduce that $\varrho_n$ converges weakly to $\varrho_*$ in $L^1(\Omega \times I_L)$ (recall that $\left( L^1 \right)^* = L^\infty$ and that $L^1 \subsetneq (L^\infty)^*$. Recall also that the unit ball of $L^1$ is not compact for the weak topology). Taking $\Phi = 1$ shows that
    \begin{equation*}
    \tau(\gamma_*) = \int_{\Omega \times I_L} \varrho_* = \lambda.
    \end{equation*}
    
    \medskip
    
    On the other hand, recall that the sequence $\tau( - \Delta_0 \gamma_n)$ is bounded.  Using Lemma~\ref{lem:inequalities}, we obtain that $(\varrho_n)$ is a bounded sequence in $L^q(\Omega \times I_L)$ for all $1 \le q \le 5/3$. In particular, up to another (undisplayed) extraction, we may assume that $(\varrho_n)$ converges weakly to some $\widetilde{\varrho}_*$ in $L^q(\Omega \times I_L)$ for all $1 < q \le 5/3$. 
    Actually, we have $\widetilde{\varrho}_* = \varrho_*$. Indeed, consider any $\varphi \in L^\infty(\Omega \times I_L)$. Then $\varphi \in L^q(\Omega \times I_L)$ as well, by Hölder inequality (the set $\Omega \times I_L$ is bounded), so $
    \int_{\Omega \times I_L} \varphi \rho_n$
    converges to both $\int_{\Omega \times I_L} \varphi \rho_*$ and $\int_{\Omega \times I_L} \varphi \widetilde{\rho_*}$. We deduce that $\int \varphi \rho_* = \int \varphi \widetilde{\rho_*}$ for all such $\varphi$. So $\widetilde{\varrho}_* = \varrho_*$.
    
    \medskip
    
    \noindent \underline{Step 5: Convergence of the energy.}
    At this point, we proved the following convergences:
    \begin{equation*}
    \begin{cases}
        \gamma_n \rightharpoonup \gamma_* \quad \text{weakly in $L^p(\cA, \tau)$ for all $1 < p < \infty$}, \\
        \gamma_n \rightharpoonup_* \gamma_* \quad \text{weakly-$*$ in $L^\infty(\cA, \tau)$}, \\
        \varrho_n \rightharpoonup \varrho_* \quad \text{weakly in $L^p(\Omega \times I_L)$ for all $1 < p \le 5/3$}.
    \end{cases}
    \end{equation*}
    In addition, we have $\tau(\gamma_*) = \lambda$ and 
    \begin{equation*}
    \tau ( - \Delta_0 \gamma_*) \le \liminf_{n \to \infty} \tau(- \Delta_0 \gamma_n).
    \end{equation*}
    According to Lemma~\ref{lem:Derg_sci},  the map $\varrho \mapsto \cD_{\rm erg}(\varrho - \fm , \varrho - \fm)$ is weakly lower-semicontinuous in $L^{5/3}(\Omega \times I_L)$. 
    We directly deduce that
    \begin{equation*}
    \cE^{\rm rHF}(\gamma_*) \le \liminf_{n \to \infty} \cE^{\rm rHF}(\gamma_n).
    \end{equation*}
    In addition, since $\gamma_* \in \cK$ with $\tau(\gamma_*) = \lambda$, we get that $\gamma_*$ is a minimizer of~\eqref{eq:min_rHF_quasi-periodic}.

    \medskip
    
    \noindent \underline{Step 6: Uniqueness of the ground-state density.}
    The kinetic term $\tau(- \Delta_0 \gamma)$ is linear in $\gamma$, while the Hartree term $\gamma \mapsto  \cD_{\rm erg}(\varrho_\gamma - \fm, \varrho_\gamma - \fm)$ is convex in $\gamma$ and strictly convex in the density. Hence all minimizers must share the same density, which we denote by $\varrho_*$.

    \medskip
    
    \noindent \underline{Step 7: Euler--Lagrange equations.} The set $\{ \gamma \in \cK, \ \tau(\gamma) = \lambda, \ \cD_{\rm erg}(\varrho_\gamma, \varrho_\gamma) < \infty \}$ is convex. For any $\gamma$ in this set with density $\varrho$, and any $0 \le t \le 1$, we have
    \begin{align*}
        \cE^{\rm rHF}(\gamma_*) & \le \cE^{\rm rHF} ((1 - t) \gamma_* + t \gamma) \\
        & \le \cE^{\rm rHF}(\gamma_*) + t \left[\frac{1}{2} \tau( - \Delta_0 (\gamma - \gamma_*)) + \cD_{\rm erg} (\varrho_* - \fm, \varrho - \varrho_*) + v_0 \fint_{\Omega} \int_{I_L} g_0 (\varrho - \varrho_*) \right] \\
        & \qquad + \frac{t^2}{2}\cD_{\rm erg} (\varrho - \varrho_*, \varrho - \varrho_*).
    \end{align*}
    Taking the limit $t \to 0^+$ shows that for all such $\gamma$, we have
    \begin{equation*}
    \frac12\tau( - \Delta_0 \gamma) + \cD_{\rm erg} (\varrho_* - \fm, \varrho) 
    + v_0 \fint_{\Omega} \int_{I_L} g_0 \varrho \ge 
    \frac12 \tau( - \Delta_0 \gamma_*) + \cD_{\rm erg} (\varrho_* - \fm, \varrho_*) + v_0 \fint_{\Omega} \int_{I_L} g_0 \varrho_*.
    \end{equation*}
    In other words, $\gamma_*$ is also the minimizer of the linearized problem
    \begin{equation} \label{eq:inf_linearized}
    \min \left\{ \frac12\tau( - \Delta_0 \gamma) + \cD_{\rm erg} (\varrho_* - \fm, \varrho_\gamma) + v_0 \fint_{\Omega} \int_{I_L} g_0 \varrho_\gamma \; ; \; \gamma \in \cK, \ \tau(\gamma) = \lambda   \right\}.
    \end{equation}
    
    \medskip
    
    Recall that $\varrho_*$ is independent of the minimizer. Let us set $\rho_{*, \omega} := \rho_*(\omega, x,z) := \varrho_*(\alpha_x \omega,z)$, and let us introduce the mean-field operator
    \begin{equation*}
    H_* := \int_\Omega H_{*, \omega} \rd \omega, \quad \text{with} \quad
    H_{*, \omega} := - \frac12\Delta_0 + V_{*, \omega} + v_0 g_0, \quad \text{with} \quad V_{*, \omega}  := V_{\rho_{*,\omega}} - V_{ m_\omega}.
    \end{equation*}
    Since $\fm \in L^{\infty,p}$ with $p > 3/2$ and $\varrho_{*} \in L^{5/3}$, we infer from Theorem~\ref{th:potential_Lrp} that the generating functions $\cV_{\fm}$ and $\cV_{\varrho_\star}$ of the stationary potentials $V_{m_\omega}$ and $V_{\rho_{*,\omega}}$ belong to $L^{\infty}(\Omega \times I_L)$ and $L^{5/3,\infty}$, respectively. It follows that for almost all $\omega$, $V_{m_\omega}$ is nonnegative and essentially bounded, while $V_{\rho_{*,\omega}}$ is nonnegative and in $L^1_{\rm loc}(S_L)$.  In particular, $0\le \cV_*^-\le\cV_\fm^+$, and $V_{*,\omega}^-\in L^\infty(S_L)$ with uniform bound in $\omega$:
\begin{equation}\label{eq:bound_potential_atomic}
    \|V_{*,\omega}^-\|_{L^\infty(S_L)} \le C_{\infty,p, L}^{\rm stat} \| \fm \|_{L^{\infty,p}}.
\end{equation} 
As a consequence, $H_{*, \omega}$ is a well-defined bounded-from-below self-adjoint operator on $L^2(S_L)$~\cite{Lew-24}. In addition,
    $V_{*,\omega}(x,z)$ is a stationary function by Theorem~\ref{th:potential_Lrp}, so that the operator $H_*$ is a fibered stationary (unbounded) self-adjoint operator affiliated to $\cA$, since $(H_*+i)^{-1}\in\cA$ and the map $\omega\mapsto H_{*,\omega}$ is measurable.
    
    \medskip
    
    Since the density of any $\gamma \in \cK$ satisfies $\varrho_\gamma \in L^{5/3}(\Omega\times I_L)$ by the Lieb--Thirring inequality~\eqref{eq:Lieb-Thirring}, $\cV_{\fm} \in L^\infty(\Omega \times I_L)$, $\cV_{\varrho_*} \in L^{5/3,\infty}$, and $\cV_{\fm},\cV_{\varrho_*} \ge 0$,  we deduce from Lemma~\ref{lem:Derg_sci} that
    \begin{align*}
    \cD_{\rm erg} (\varrho_* - \fm, \varrho_\gamma) &= \cD_{\rm erg} (\varrho_*, \varrho_\gamma) - \cD_{\rm erg} (\fm, \varrho_\gamma) =  \int_{\Omega \times I_L} \cV_{\varrho_*} \varrho_\gamma -  \int_{\Omega \times I_L} \cV_{\fm} \varrho_\gamma \\ &= \tau(\cV_{\varrho_*} \gamma) - \tau(\cV_{\fm} \gamma)= \tau(\cV_* \gamma).
    \end{align*}
    So the linearized problem~\eqref{eq:inf_linearized} also reads
    \begin{equation}\label{eq:linearized_min_pb_rhf}
    \min  \left\{ \tau( H_* \gamma ), \; \gamma \in \cK, \ \tau(\gamma) = \lambda   \right\}.
    \end{equation}

\begin{lemma}\label{lem:EulerLagrange}
For all $\mu \in \R$, the projector $\Pi_\mu  := \1_{(-\infty,\mu)}(H_*)$ is a fibered stationary operator in $W_0^{1,1}(\cA,\tau)$ and all the minimizers of the problem
\begin{equation} \label{eq:min_rHF_fixed_potential}
\inf \left\{ \tau(H_*\gamma)  - \mu \tau(\gamma), \; \gamma \in \mathcal K \right\},
\end{equation}
are of the form
$$
\gamma_\mu = \Pi_\mu + \delta, \quad \mbox{with} \quad \delta \in  W_0^{1,1}(\cA,\tau) \text{ self-adjoint}, \quad 0 \le \delta \le \1_{\{\mu\}}(H_{*}).
$$
\end{lemma}

\begin{proof}
Let us show that $\Pi_\mu$ belongs to $\cK$. First, $\Pi_\mu$ is a self-adjoint fibered stationary operator and $0\le \Pi_\mu\le 1$, which is clear since the function $0\le \1_{(-\infty,\mu)}\le 1$ is measurable and real-valued, and $H_*$ is affiliated with $\cA$. We then have to show that $\Pi_\mu\in W_0^{1,1}(\cA,\tau)$. We claim that, for a.e. $\omega\in\Omega$, we have
\begin{equation} \label{eq:bound_rho_A_omega}
    \rho_{A_\omega} \leq C e^{(\|V_{*,\omega}^-\|_{L^\infty(S_L)} + | v_0|/2)}, \quad \mbox{where} \quad A_\omega:=e^{-H_{*,\omega}}=e^{\frac12\Delta_0-V_{*,\omega} - v_0 g_0}.
\end{equation}
This follows from the Feynman-Kac formula for the Dirichlet Laplacian~\cite{simon1979functional}, and uses the fact that the heat kernel of $-\Delta_0$ on $S_L$ is pointwise bounded by the heat kernel of $-\Delta$ on $\R^3$. The remainder of the argument is then identical to~\cite[Prop.~2.11]{cances2013meanfield}.

The right-hand-side of the inequality in~\eqref{eq:bound_rho_A_omega} is finite and bounded uniformly in $\omega$, by~\eqref{eq:bound_potential_atomic}.
Using the inequality $\1_{(-\infty,\mu)}(x)\leq e^{\mu-x}$, we get that $\varrho_{\Pi_\mu}\in L^\infty(\Omega\times I_L)\subset L^1(\Omega\times I_L)$, so that $\Pi_\mu\in L^1(\cA,\tau)$:
\begin{equation}\label{eq:Feynman-Kac}
    0\leq \varrho_{\Pi_\mu}\leq C_{p,\fm,L}e^\mu,
\end{equation}
where the constant $C_{p,\fm,L}>0$ comes from~\eqref{eq:bound_potential_atomic}.
Furthermore, the inequality $x\1_{(-\infty,\mu)}(x)\leq\mu\1_{\mu\ge 0}e^{\mu-x}$ combined to the fact that $H_*\ge -C$ yields
$$-C\Pi_\mu\leq H_*\Pi_\mu\leq \mu\1_{\mu\ge 0}\, e^{-(H_*-\mu)}.$$
Hence $H_*\Pi_\mu\in L^1(\cA,\tau)$, which implies that $\Pi_\mu\in W_0^{1,1}(\cA,\tau)$.

\medskip 

Finally, adapting~\cite[Prop. 2.12]{cances2013meanfield}, we show that the minimizers of~\eqref{eq:min_rHF_fixed_potential} are of the stated form. Let $\gamma\in\cK$. Since $0\leq\gamma\leq 1$, we get that
\begin{equation*}
    \Pi_\mu^\perp(\gamma-\Pi_\mu)\Pi_\mu^\perp-\Pi_\mu(\gamma-\Pi_\mu)\Pi_\mu \geq (\gamma-\Pi_\mu)^2,
\end{equation*}
yielding formally that 
\begin{equation*}
    \tau(\Pi_\mu^\perp|H_*-\mu|\Pi_\mu^\perp(\gamma-\Pi_\mu))-\tau(\Pi_\mu|H_*-\mu|\Pi_\mu(\gamma-\Pi_\mu))\geq \tau(|H_*-\mu|(\gamma-\Pi_\mu)^2).
\end{equation*}
Now, the left-hand-side, using the definition of $\Pi_\mu$, is equal to $\tau((H_*-\mu)(\gamma-\Pi_\mu))$, leading to
\begin{equation}\label{eq:ineq_functional_projectors}
    \tau((H_*-\mu)(\gamma-\Pi_\mu))\geq \tau(|H_*-\mu|(\gamma-\Pi_\mu)^2).
\end{equation}
Define the functional
\begin{equation}\label{eq:grand_canonical_energy_func}
    \cF(\gamma):=\frac{1}{2}\tau(-\Delta_0\gamma)+\int_{\Omega\times I_L}(\cV_* + v_0 g_0) \varrho_\gamma -\mu\tau(\gamma).
\end{equation}
As mentioned above, the regularity of $\cV_*$ allows to write that
\begin{equation}\label{eq:rep_func_trace}
    \cF(\gamma)=\tau((H_*-\mu)\gamma).
\end{equation}
Then,~\eqref{eq:ineq_functional_projectors} makes sense, and it yields for $\gamma\in\cK$, that 
\begin{equation}\label{eq:minimality_projectors}
    \cF(\gamma)-\cF(\Pi_\mu) \geq \tau(|H_*-\mu|^{1/2}(\gamma-\Pi_\mu)^2|H_*-\mu|^{1/2})\geq 0.
\end{equation}
Hence $\Pi_\mu$ is clearly a minimizer of~\eqref{eq:linearized_min_pb_rhf}, and every minimizer $\gamma$ must verify that $|H_*-\mu|^{1/2}(\gamma-\Pi_\mu)=0$, so that $\mathrm{Ran}(\gamma-\Pi_\mu)\subset\mathrm{Ker}(H_*-\mu)$.
\end{proof}

To conclude the proof of the theorem, we proceed as in~\cite[Prop. 4.5]{cances2013meanfield}. As mentioned above, $\gamma_*$ is a minimizer of the linearized problem~\eqref{eq:linearized_min_pb_rhf}, that is, for $\gamma\in\cK$ such that $\tau(\gamma)=\lambda$,
\begin{equation*}
    \tau(H_*(\gamma-\gamma_*))\geq 0.
\end{equation*}
Define, for $\mu \in \R$,
\begin{equation*}
    E(\mu) := \inf \left\{\tau(H_*(\gamma-\gamma_*)), \; \gamma\in\cK,\ \tau(\gamma)=\mu\right\}.
\end{equation*}
The function $E:\R \to \R$ is convex, hence continuous and left and right differentiable at each $\mu \in \R$. Choose any 
\begin{equation*}
    \eps_F\in [E'(\lambda^-),E'(\lambda^+)],
\end{equation*}
where $E'(\lambda^-)$ and $E'(\lambda^+)$ respectively denote the left and right derivatives of $E$ at $\lambda$. Using the fact that $E(\lambda)=0$, we have that for any $\gamma\in\cK$,
\begin{equation*}
    \eps_F\tau(\gamma-\gamma_*)=\eps_F(\tau(\gamma)-\lambda) \leq E(\tau(\gamma))-E(\lambda) = E(\tau(\gamma)) \leq \tau(H_*(\gamma-\gamma_*)),
\end{equation*}
which gives $\tau((H_*-\eps_F)(\gamma-\gamma_*))\geq 0$. It follows that $\gamma_*$ is also a minimizer of~\eqref{eq:min_rHF_fixed_potential} with $\mu=\eps_F$, so that, by Lemma~\ref{lem:EulerLagrange}, $\gamma_*$ is of the form~\eqref{eq:characterization_minimizers}.

\medskip
Finally, since any minimizer $\gamma_*$ satisfies $\gamma_* \le \Pi_{\eps_F+1}$, we have $\varrho_* \le  \varrho_{\Pi_{\eps_F+1}}$, which, in view of the bound~\eqref{eq:Feynman-Kac}, implies that $\varrho_* \in L^\infty(\Omega\times I_L)$. It follows from Theorem~\ref{th:potential_Lrp} that the potential $\cV_*$ is bounded as well.
This concludes the proof of Theorem~\ref{th:main_KS_moire}.
\end{proof}

\section{Kohn--Sham models for encapsulated periodic 2D materials}
\label{sec:periodic}
    
For completeness, we provide a proof of existence of a Kohn-Sham LDA  ground state for encapsulated 2D periodic materials. Unfortunately, our proof does not extend to the quasi-periodic setting. The most usual and practical way to write down a periodic Kohn--Sham model is to rely on Bloch theory. In order to recycle what has been done in the previous section, we will however follow a different (but equivalent) approach consisting in reformulating the periodic problem as a stationary ergodic problem.

\subsection{The periodic setting as a stationary one}
Consider a 2D material with $\L$-periodic nuclear charge distribution $m : S_L \to \R$, where $\L$ is a periodic lattice of $\R^2$. We introduce the configuration space $\widetilde{\Omega} := \R^2 / \L $, and the group action  $\R^2 \ni x \mapsto \widetilde{\alpha}_x$, defined by
\begin{equation*}        \widetilde{\alpha}_x : \widetilde{\Omega} \to \widetilde{\Omega}, \qquad \widetilde{\alpha}_x(\widetilde \omega)  = \widetilde \omega - x  \mod(\L).
\end{equation*}
For a function $f:S_L \to \R$, periodic in the sense $f(x + \ell, z) = f(x, z)$ for all $\ell \in \L $, we introduce the function $\widetilde f : \widetilde{\Omega} \times \R^2 \times I_L \to \R$ defined by
$$
\widetilde f(\widetilde\omega, x,z) := \widetilde f_{\widetilde\omega}(x,z) := f(x-\widetilde\omega,z).
$$
Then $\widetilde{f}$ is stationary for the group action $\widetilde\alpha$. In this context, the configuration $\widetilde \omega \in \widetilde \Omega$ encodes a global shift of the lattice $\L $.  Note that $\widetilde f_{\widetilde{\omega}=0} = f$ and that the generating function $\widetilde{\ff}$ of $\widetilde f$ can be seen as the original function $f$, up to inversion:
    $$
    \forall \widetilde\omega \in \widetilde\Omega, \quad \widetilde{\ff}(\omega,z)=f(-\omega,z).
    $$
    This proves that periodic functions can be identified with stationary function, with this specific group action.
    
    The action $\widetilde\alpha$ is ergodic (the only functions invariant by all shifts are the constant functions). Actually, Birkhoff's theorem becomes trivial in this special case: for any configuration $\widetilde\omega$, and any locally-integrable $\L$-periodic function $f$, we have with the above notation 
    \begin{align*}
    & \lim_{R \to \infty} \frac{1}{(2R)^2} \int_{[-R,R]^2 \times I_L} \widetilde f_{\widetilde \omega}(x,z) \, \rd x \, \rd z = 
    \lim_{R \to \infty} \frac{1}{(2R)^2} \int_{[-R,R]^2 \times I_L} f(x-\widetilde\omega,z) \, \rd x \, \rd z \\
    & \qquad = 
    \lim_{R \to \infty} \frac{1}{(2R)^2} \int_{([-R,R]^2-\widetilde\omega) \times I_L} f(x,z) \, dx
      = \int_{I_L}  \fint_{\R^2/\L} f(x,z) \, \rd x \, \rd z  = \int_{I_L} \fint_{\R^2/\L} \widetilde {\ff}(\omega,z) \, \rd \omega \rd z.
    \end{align*}
    Using this observation, the original periodic Kohn--Sham model for the nuclear charge distribution $m$, can be transformed into a stationary Kohn--Sham problem with stationary density $\widetilde m$. 
    
    \medskip
    
    Before writing down the corresponding encapsulated LDA model, let us emphasize the fundamental difference between the stationary ergodic settings corresponding to periodic and quasi-periodic 2D materials respectively. To illustrate this point, consider the simple case of twisted bilayer graphene (TBG) with twist angle $\theta$. We define the configuration space $\Omega$ as $\Omega = \R^2/\L_1 \times \R^2/\L_2$, with $\L_1:=R_{-\theta/2}\L_{\rm mg}$, $\L_2:=R_{\theta/2}\L_{\rm mg}$, where $\L_{\rm mg}$ is the monolayer graphene triangular lattice, and $R_{\varphi}$ the 2D rotation matrix of angle $\varphi$, and the action $\alpha$ of $\R^2$ on $\Omega$  by $\alpha_x(\omega_1,\omega_2):=(\omega_1-x \mod(\L_1),\omega_2-x \mod(\L_2))$. Recall that the system is in the configuration $\omega=(\omega_1,\omega_2) \in \Omega$ if the first graphene layer has a carbon atom of type A at $(\omega_1,z_1) \in S_L$ and the second layer a carbon atom of type A at $(\omega_2,z_2) \in S_L$, where $z_j$ is the height of the plane containing the $j^{\rm th}$ monolayer. The action $\alpha$ encodes uniform longitudinal translations of the system.
    Let 
    $$
    \Theta_{\rm per}:=\left\{ \theta \in [0,\pi/3] : \; \cos\theta = \frac{3 p^2 - q^2}{3p^2+q^2}, \quad p,q \in \N^*, \quad  p \wedge q = 1 \right\} \cup \{ 0 \}.
    $$
    If $\theta \in \Theta_{\rm per}$, then the system is $\L$-periodic, with $\L:=\L_1 \cap \L_2$. If $\theta \in [0,\pi/3] \setminus \Theta_{\rm per}$, it is quasi-periodic (the lattices are incommensurate).
    If $\theta \in [0,\pi/3]\setminus \Theta_{\rm per}$, $\alpha$ is an ergodic action and all configurations have the same energy. If $\theta \in \Theta_{\rm per}$, then $\alpha$ is not ergodic and the energy depends on the configuration~$\omega$. Of course, if $\omega_0$ is the configuration corresponding to the periodic charge distribution $m$ and $O_{\omega_0}:=\{\alpha_x\omega_0, \, x \in \R^2\}$ the orbit of $\omega_0$, all the configurations in $O_{\omega_0}$ have the same energy per unit area.
    These are precisely the configurations sampled by the action of $\widetilde{\alpha}$ on $\widetilde \Omega$. In the periodic setting, all the configurations in $\widetilde \Omega$ are strictly identical from a physical viewpoint (they are mere translations of one another), while in the quasi-periodic setting, all the configurations in $\Omega$ are not identical (only the ones in a given orbit $O_{\omega_0}$ are identical), although they have the same energy per unit area.
    
    \subsection{The periodic Kohn--Sham in the stationary framework} That being said, let us now formulate the encapsulated Kohn--Sham LDA model for a 2D material with $\L$-periodic nuclear charge distribution $m$. 
    We make the assumption that the ground-state one-body density operator (and its corresponding density) is also $\L $-periodic, i.e. that there is no symmetry breaking although this assumption may fail due to the non convexity of the problem, see~\cite[Section 3]{gontier2021nonlinear} (following~\cite{ricaud2018symmetry}). 
    
    \medskip
    
    It is easy to see that the von Neumann algebra $\mathcal A$ associated with $(S_L,\widetilde\Omega,\widetilde \alpha)$ is isometric to the von Neumann algebra of $\L$-periodic bounded operator on $L^2(S_L)$, and that the trace $\widetilde\tau$ on $\widetilde \cA$ agrees (for a suitable scaling factor) with the trace per unit area $\underline{\rm Tr}$ on the space of $\L$-periodic operators on $L^2(S_L)$. More precisely, if $A$ is a locally-trace-class bounded $\L$-periodic operator on $L^2(S_L)$ and $\widetilde A$ its counterpart in $\widetilde\cA$, it holds that $\widetilde A \in L^1(\widetilde \cA,\widetilde\tau)$ and we have
    $$
    \widetilde\tau(\widetilde A) = \frac{1}{|\R^2/\L|} \fint_{\R^2/\L^*} {\rm Tr}_{L^2(\R^2/\L \times I_L)} \left( A_k \right) \, \rd k = \underline{{\rm Tr}}(A),
    $$
    where 
    $$
    A = \fint_{\R^2/\L^*}^\oplus A_k \, \rd k
    $$
    is the Bloch decomposition of $A$. Likewise, the ergodic Hartree energy per unit area $\widetilde D_{\rm erg}$ of the generating function $\widetilde {\ff}$ is related to the periodic Hartree energy per unit cell of the periodic charge distribution $f$ (see~\eqref{eq:def:Hartree_Periodic}) by
    $$
    \widetilde D_{\rm erg}(\widetilde {\ff},\widetilde {\ff})= \frac{1}{|\R^2/\L|}  D_{\rm per}(f,f).
    $$
    Finally, the LDA exchange-correlation energy per unit area of a periodic electronic density $\rho$, or equivalently of the generating function $\widetilde{\varrho}$ since $\widetilde{\varrho}(\omega,z)=\rho(-\omega,z)$, is given by 
     \begin{equation*}
    E^{\rm xc}(\rho) :=  E^{\rm xc}(\widetilde\varrho) = \int_{I_L} \fint_{(\R^2/\L)}  \re^{\rm xc} \big( \widetilde\varrho(\omega, z) \big) \rd \omega \rd z,
    \end{equation*}
    where the function $\re^{\rm xc}:\R_+ \to \R_-$ is the exchange-correlation energy per unit volume of the homogeneous electron gas (or an approximation of this function). Recall that the Local Density Approximation (LDA) was first introduced by Kohn and Sham~\cite{kohn1965selfconsistent}, and was recently justified by rigorous mathematical arguments in some regime~\cite{lewin2019local}. In what follows, we assume that the function $\re^{\rm xc} : \R^+ \to \R^-$ is continuous and satisfies the bound
    \begin{equation} \label{eq:Assumption_1}
        \forall r \ge 0, \qquad 0 \ge \re^{\rm xc}(r) \ge - C (r^{\beta_1} + r^{\beta_2} ), \qquad 1 \le \beta_1 \le \beta_2  <  \frac53.
    \end{equation}
    This assumption is satisfied for the Dirac exchange term $\re^{\rm xc}(r) = - c_D r^{4/3}$, as well as for most functionals used in practice~\cite{perdew1981selfinteraction, perdew1992accurate}.
    
    \medskip
    
    We finally define the Kohn-Sham energy functional of the $\L$-periodic density matrix $\gamma$ as (we set $v_0 = 0$ for simplicity)
    \begin{align*}
         \cE^{\rm KS}_{\rm per} (\gamma) := & \frac{1}{2|\R^2/\L|} \fint_{\R^2/\L^*} {\rm Tr}_{L^2(\R^2/\L \times I_L)} \left( \left( (-i\nabla_x+k)^2-\partial_{zz,0} \right) \gamma_k \right) \, dk   \\ & + \frac1{2|\R^2/\L|} \cD_{\rm per}(\rho_\gamma - m, \rho_\gamma - m) + E^{\rm xc}(\rho_\gamma),
    \end{align*}
    or equivalently, in terms of the stationary operator $\widetilde\gamma$ associated with $\gamma$, as
    \begin{equation*}
         \boxed{ \widetilde \cE^{\rm KS}_{\rm per} (\widetilde \gamma) := \frac12 \widetilde\tau( - \Delta_0 \widetilde\gamma) + \frac12 \cD_{\rm per}(\widetilde\varrho_{\widetilde\gamma} - \widetilde\fm, \widetilde\varrho_{\widetilde\gamma} - \widetilde\fm) + E^{\rm xc}(\widetilde\varrho_\gamma).}
    \end{equation*}

    \begin{theorem} \label{th:main_KS_moire_per}
    For all $\lambda > 0$, the encapsulated periodic model
    \begin{equation*}
    \inf \left\{ \cE^{\rm KS}_{\rm per}(\gamma) , \quad \gamma \in \mathcal B(L^2(S_L)), \quad 0 \le  \gamma=\gamma^* \le 1,  \ \gamma \mbox{ $\L$-periodic}, \ \underline{\rm Tr}(\gamma) = \lambda, \  \underline{\rm Tr}(-\Delta_0\gamma) < \infty \right\},
    \end{equation*}
    or equivalently
    \begin{equation}\label{eq:minimization_pb_periodic}
    \inf \left\{ \widetilde\cE^{\rm KS}_{\rm per}(\widetilde\gamma) , \quad \widetilde\gamma \in W_0^{1,1}(\widetilde\cA,\widetilde\tau), \quad 0 \le \widetilde \gamma=\widetilde\gamma^* \le 1,  \ \widetilde\tau(\widetilde\gamma) = \lambda \right\}.
    \end{equation}
    has a ground state. 
\end{theorem}
    
This time, the energy is no longer convex, due to the LDA term, so do not expect uniqueness in general.

\begin{proof}
    We only highlight the main differences with the previous proof, focusing on the minimization problem~\eqref{eq:minimization_pb_periodic}. For the sake of readability, we drop the tilde notations in the proof.
    
    \noindent\underline{Step 1} (the minimization set is non empty) is similar. We consider $\gamma$ of the form~\eqref{eq:example_non_empty}, with $h$ continuous, then $\varrho$ is continuous on $\Omega \times I_L$, and $E^{\rm xc}(\rho)$ is bounded. 
    \smallskip
    
    \noindent \underline{Step 2 bis. The energy is bounded from below.} The LDA term is negative, and we need to control it. Using hypothesis~\eqref{eq:Assumption_1}, we get
    \begin{equation*}
        E^{\rm xc} (\varrho) \ge - C \left( \|  \varrho \|_{\beta_1}^{\beta_1} + \|  \varrho \|_{\beta_2}^{\beta_2} \right).
    \end{equation*}
    Since $1 \le \beta_1 < 5/3$, we have the Hölder's inequality
    \begin{equation*}
        \| \varrho \|_{\beta_1}^{\beta_1} \le \| \varrho \|_1^{\theta_1 \beta_1} \| \varrho \|_{5/3}^{(1 - \theta_1) \beta_1 }, \qquad 
        \text{with $0 < \theta_1 \le 1$ chosen so that} \quad \theta_1 + \frac{3(1 - \theta_1)}{5} = \frac{1}{\beta_1}.
    \end{equation*}
    Together with the fact that $\| \varrho \|_1 = \lambda$ and Lieb--Thirring inequality, we obtain
    \begin{equation*}
        \| \varrho \|_{\beta_1}^{\beta_1} \le C \lambda^{\theta_1 \beta_1} \, \tau( - \Delta_0 \gamma)^{\frac{3(1 - \theta_1) \beta_1}{5}} = C \lambda^{\theta_1 \beta_1} \, \tau( - \Delta_0 \gamma)^{1 - \theta_1 \beta_1}.
    \end{equation*}
    We have a similar computation for $\beta_2$, with some $0<\theta_2\le 1$. This proves that
    \begin{equation*}
        \cE^{\rm KS}(\gamma) \ge \frac12\tau( - \Delta_0 \gamma) - C_\lambda \left[ \tau( - \Delta_0 \gamma)^{1 - \theta_1 \beta_1} + \tau( - \Delta_0 \gamma)^{1 - \theta_2 \beta_2}\right].
    \end{equation*}
    The function $X \mapsto X - C_\lambda (X^{1 - \theta_1 \beta_1} + X^{1 - \theta_2 \beta_2})$ is coercive. We deduce first that the energy is bounded from below, and that any minimizing sequence $\gamma_n$ has $\tau( - \Delta_0 \gamma_n)$ bounded.
    
    \medskip
    
    \noindent \underline{Step 3} and \underline{Step 4} are similar as before (extraction and convergence of a subsequence). From the periodic Hoffmann--Ostenhof inequality
    \begin{equation} \label{eq:periodic_HO}
        \tau (- \Delta_0 \gamma) \ge \int_{I_L} \fint_{\R^2 / \L } | \nabla \sqrt{\varrho} |^2 \rd \omega \rd z,
    \end{equation}
    we also deduce that the sequence $\nabla \sqrt{\varrho_n}$ is bounded in $L^2(\Omega \times I_L)$. Together with Rellich---Kondrachov embedding (compact Sobolev embedding), we deduce that $\varrho_n$ converges {\bf strongly} to $\varrho_*$ in $L^p(\Omega \times I_L)$ for all $1 \le p < 3$, and almost everywhere up to a subsequence (we used the fact that the set $\Omega \times I_L$ is compact). \\
    
    \medskip
    
     \underline{Step 5 bis. Convergence of the LDA energy}. Reasoning as before, we can prove that
     \begin{equation*}
        \tau(- \Delta_0 \gamma_*)  \le \liminf_{n \to \infty} \tau( - \Delta_0 \gamma_n) \quad \text{and} \quad
        \cD_{\rm per} (\varrho_* - \fm, \varrho_* - \fm) \le \liminf_{n \to \infty } \cD_{\rm per} (\varrho_n - \fm, \varrho_n - \fm).
     \end{equation*}
    It remains to prove the convergence of the LDA term. We write
    \begin{equation*}
        \re^{\rm xc}(r) := - C (r^{\beta_1} + r^{\beta_2}) + f(r), \quad \text{with} \quad f(r) \ge 0.
    \end{equation*}
    Since $\re^{\rm xc}$ is continuous, so is $f$. This gives
    \begin{equation*}
        \re^{\rm xc}(\varrho_n) = - C ( \varrho_n^{\beta_1} + \varrho_n^{\beta_2} ) + f(\varrho_n).
    \end{equation*}
    From the strong convergence of $\varrho_n$ to $\varrho_*$ in $L^p(\Omega \times I_L)$ for all $1 \le p < 3$, we get
    \begin{equation*}
        \int_{I_L} \fint_\Omega \varrho_n^{\beta_1} + \varrho_n^{\beta_2} \xrightarrow[n \to \infty]{} 
        \int_{I_L} \fint_\Omega \varrho_*^{\beta_1} + \varrho_*^{\beta_2} .
    \end{equation*}
    On the other hand, $\varrho_n$ converges pointwise to $\varrho_*$ a.e. By continuity of $f$, we deduce that $f(\varrho_n)$ pointwise to $f(\varrho_*)$ a.e. In addition, $f$ is positive, so Fatou's lemma implies that
    \begin{equation*}
        \int_{I_L} \fint_\Omega f(\varrho_*) \le \liminf_{n \to \infty }\int_{I_L} \fint_\Omega f(\varrho_n).
    \end{equation*}
    Altogether, we proved that
    \begin{equation*}
        E^{\rm xc}(\gamma_*) \le \liminf_{n \to \infty }E^{\rm xc}(\gamma_n), \quad \text{hence} \quad 
        \cE^{\rm KS}(\gamma_*) \le \liminf_{n \to \infty } \cE^{\rm KS}(\gamma_n).
    \end{equation*}
    This proves as before that $\gamma_*$ is a minimizer of the Kohn--Sham energy $\cE^{\rm KS}$.
\end{proof}

\begin{remark}
    The additional property that we have in the periodic setting is the strong convergence of $\widetilde{\varrho}_n$ to $\widetilde{\varrho}_*$ in $L^p(\widetilde{\Omega} \times I_L)$, which allows to prove the convergence of the LDA term. This strong convergence comes from the Hoffmann--Ostenhof inequality~\eqref{eq:periodic_HO} involving the $\nabla$ term, which is such that  $\nabla^* \nabla $ is coercive. In the quasi-periodic setting, the same inequality involves $\Lambda$, and $\Lambda^* \Lambda$ is not coercive. There is no equivalent of the compact Sobolev embedding with the operator $\Lambda$.
\end{remark}
    
    \section*{Acknowledgements} 
    
    This project has received funding from the European Research Council (ERC) under the European Union's Horizon 2020 research and innovation program (grant agreement EMC2 No 810367), the Simons Targeted Grant in Mathematics and Physical Sciences on Moir\'e Materials Magic (Award No. 896630, E.C.). This work has benefited from French State support managed by ANR under the France 2030 program through the MaQui CNRS Risky and High-Impact Research program (RI)$^2$ (grant agreement ANR-24-RRII-0001). 
    The authors are grateful to Huajie Chen, Antoine Levitt, Mathieu Lewin, Amine Marrakchi, Basile Morando, and Reinhold Schneider for useful discussions. 

    \section*{Use of generative AI}
    ChatGPT Astra and Claude Opus 5.5 were used to check the proofs, improve the bibliography, and correct typos. No AI tools were used to design the model, nor to produce mathematical results, nor to write the manuscript.

    \section*{Data availability}
    Data sharing is not applicable to this article as no data sets were generated or analyzed during the current study.

    \section*{Conflict of interest}
    The authors declare that they have no conflict of interest.

\printbibliography
\end{document}